\pdfoutput=1
\documentclass[11pt,letterpaper]{article}

\usepackage{style}
\usepackage{shortcuts}

\crefname{observation}{Observation}{Observations}
\Crefname{observation}{Observation}{Observations}
\crefname{claim}{Claim}{Claims}
\Crefname{claim}{Claim}{Claims}
\crefname{openproblem}{Open Problem}{Open Problems}
\Crefname{openproblem}{Open Problem}{Open Problems}

\title{Sumsets, 3SUM, Subset Sum: Now for Real!}
\author{Nick Fischer\thanks{INSAIT, Sofia University ``St. Kliment Ohridski''. Partially funded by the Ministry of Education and Science of Bulgaria's support for INSAIT, Sofia University ``St. Kliment Ohridski'' as part of the Bulgarian National Roadmap for Research Infrastructure. Parts of this work were done while the author was at Weizmann Institute of Science.}}
\date{}

\begin{document}
\maketitle

\begin{abstract}
\noindent
We study a broad class of algorithmic problems with an ``additive flavor'' such as computing sumsets, 3SUM, Subset Sum and geometric pattern matching. Our starting point is that these problems can often be solved efficiently for \emph{integers}, owed to the rich available tool set including bit-tricks, linear hashing, and the Fast Fourier Transform. However, for \emph{real numbers} these tools are not available, leading to significant gaps in the best-known running times for integer inputs versus for real inputs. In this work our goal is to close this gap.

As our key contribution we design a new technique for computing \emph{real} sumsets. It is based on a surprising blend of algebraic ideas (like \emph{Prony's method} and \emph{coprime factorizations}) with combinatorial tricks. We then apply our new algorithm to the aforementioned problems and successfully obtain, in all cases, equally fast algorithms for real inputs. Specifically, we replicate the running times of the following landmark results by randomized algorithms in the standard real RAM model:
\begin{itemize}
    \item\emph{Sumsets:} Given two sets $A, B$, their sumset $A + B = \set{a + b : a \in A, b \in B}$ can be computed in time~\smash{$\widetilde O(|A + B|)$} [Cole, Hariharan; STOC'02].
    \item\emph{Geometric pattern matching:}  Given two sets $A, B$, we can test whether there is some shift such that $A + s \subseteq B$ in time $\widetilde O(|A| + |B|)$ [Cardoze, Schulman; FOCS'98].
    \item\emph{3SUM with preprocessing:} We can preprocess three size-$n$ sets $A, B, C$ in time~\smash{$\widetilde O(n^2)$} such that upon query of sets $A' \subseteq A, B' \subseteq B, C' \subseteq C$, the 3SUM instance $(A', B', C')$ can be decided in time~\smash{$\widetilde O(n^{13/7})$} [Chan, Lewenstein; STOC'15].
    \item\emph{Output-sensitive Subset Sum:} Given a size-$n$ (multi-)set $X$ and a target $t$, we can compute the set of subset sums $\set{\Sigma(X') : X' \subseteq X, \Sigma(X') \leq t}$ in output-sensitive time~\smash{$\widetilde O(n + \OUT^{4/3})$} [Bringmann, Nakos; STOC'20].
\end{itemize}
\end{abstract}

\thispagestyle{empty}
\newpage
\setcounter{page}{1}

\section{Introduction} \label{sec:intro}
We study a broad set of basic computational problems with an ``additive component'' such as 3SUM, Subset Sum, computing sumsets and geometric pattern matching. These problems are of great importance to several communities, notably within fine-grained algorithms and complexity, computational geometry, and, to a lesser degree, also in cryptography and computer algebra, and have received substantial attention in the literature. Especially in recent years this effort resulted in many beautiful algorithms, often powered by tools from additive combinatorics.

In this work we revisit these additive problems for \emph{real} input numbers. Our general starting point is the following wide-reaching question:
\begin{center}
    \smallskip
    \emph{Can we adapt the landmark results for ``additive problems'' to real inputs?}
    \smallskip
\end{center}
This question is not only naturally fundamental, but is also driven by countless theoretical and practical applications, particularly in computational geometry, some of which we will encounter soon. It has been studied for many important algorithmic problems (see e.g.~the exposition on 3SUM algorithms in the following paragraph), but was also asked in the context of fine-grained lower bounds; see~\cite{KopelowitzP18,ChanH20} and notably~\cite{ChanWX22}.

On the positive side, the majority of integer algorithms can easily be adapted to the reals (often because the algorithm accesses the input numbers in a limited manner). However, a rapidly growing list~\cite{CardozeS98,BaranDP08,ArnoldR15,ChanL15,Roche18,Nakos20,GiorgiGC20,BringmannN20,BringmannFN21,BringmannFN22,JinX24,Fischer24} of algorithms rely on the input consisting of integers by exploiting certain integer-specific tools (such as bit-tricks), rendering some modern approaches useless beyond repair in the case of reals. In this paper our goal is to address this shortcoming, and to nevertheless find appropriate replacement algorithms matching the integer complexity.

\paragraph{Case Study: Log-Shaves for 3SUM}
As an inspiring example of this question in the literature, consider the fundamental 3SUM problem: Given a set $A$ of~$n$ numbers, the goal is to test whether there exist $a, b, c \in A$ with $a + b = c$. It is a prominent conjecture that 3SUM cannot be solved in truly subquadratic time~\cite{GajentaanO95,VassilevskaWilliams18}, and even small improvements over the naive $\Order(n^2)$-time algorithm are urgently sought after. In~\cite{BaranDP08}, Baran, Demaine and Pătraşcu developed the first mildly subquadratic algorithm for integer 3SUM, shaving nearly two log-factors from the naive quadratic time. However, their algorithm crucially relies on \emph{linear hashing}---a technique which only works for integers and is not applicable to real 3SUM. This was considered a major insufficiency as the original intention behind 3SUM (and its applications to the hardness of geometric problems) had the real-valued variant in mind~\cite{GajentaanO95}. In a breakthrough result, Gr{\o}nlund and Pettie~\cite{GronlundP18} later scored the first mildly subquadratic algorithm for real 3SUM. They inspired an entire line of follow-up research~\cite{GronlundP18,Freund17,GoldS17,Chan20} (see also~\cite{KaneLM19}). Only recently, and through the development of powerful geometric tools, Chan~\cite{Chan20} successfully replicated (roughly) the same state of affairs for real 3SUM as for integer 3SUM.

\paragraph{Integer-Specific Techniques}
For many other additive problems, however, the gap between the fastest algorithms for integers versus for reals is still not closed---and for all problems considered in this paper the gap is more dramatic than ``just'' log-factors. The reason, in all these cases, is the routine use of integer-specific techniques, at least indirectly, and often in combination with each other. Specifically, we can almost exclusively trace back the issues to the following palette of problematic techniques:

\begin{enumerate}
    \setlength\parskip{0pt}
    \item\emph{Bit-tricks, rounding, scaling:}
    The first integer-specific tool that comes to mind is \emph{bit-tricks}. A typical use case is that the computation of some $L$-bit integer (where usually $L = \Order(\log n)$) can be reduced to the $L$ decision problems of testing the respective bits. This can often be achieved by appropriately \emph{rounding} the involved numbers, and is also occasionally called \emph{scaling.}
    
    For obvious reasons, this idea cannot work for real numbers that generally require infinitely many bits. In fact, removing this scaling trick is \emph{the} obstacle in achieving strongly polynomial-time algorithms for a plethora of algorithms from combinatorial optimization and is throughout considered one of the most important questions in the area; see~\cite{Fineman24} for a recent breakthrough on single-source shortest paths without scaling.
    \item\emph{Linear hashing:}
    We call a hash function \emph{linear} if it satisfies that $h(x) + h(y) = h(x + y)$.\footnote{This at least is the moral requirement; for technical reasons many use cases of linear hashing resort to a slightly weaker condition.} Linear hashing often plays an important role to design \emph{universe reductions} for additive problems. Specifically, consider a size-$n$ set $A \subseteq \set{0, \dots, u}$ over a universe of size $u = n^{100}$. It is often inconvenient that $u$ is large (e.g.~an efficient algorithm cannot spend time proportional to~$u$). Letting $h : [u] \to [n]$ be a random linear hash function, we can replace $A$ by the dense set~\makebox{$\set{h(a) : a \in A} \subseteq [n]$} while preserving its additive properties to a certain degree.

    For reals, on the other hand, we cannot expect a similar universe reduction. Hashing the continuous interval from, say, $[0, n^{100}]$ to $[0, n]$, is typically meaningless as many problems are scale-invariant. And hashing to $[n]$ instead either breaks linearity or requires an unconventional model of computation (to support, say, flooring).
    \item\emph{Fast Fourier Transform:}
    The Fast Fourier Transform (FFT) is one of the most influential computational primitives with countless applications across various disciplines. A typical algorithmic consequence of the FFT is that the \emph{sumset} $A + B = \set{a + b : a \in A, b \in B}$ of two given sets~\makebox{$A, B \subseteq [u]$} can be computed in time $\Order(u \log u)$. As we will describe shortly, computing sumsets lies at the heart of involved algorithms for additive-type problems. Note that this paradigm is particularly efficient in combination with linear hashing (i.e., after reducing the universe size $u$).

    However, here we crucially require that the sets $A, B$ are from the discrete universe $[u]$, and there is no hope to achieve a similar algorithm for set of real numbers $A, B \in [0, u]$. (The running time $\Order(u \log u)$ does not even make sense here as $u$ can be arbitrarily small or large.)
\end{enumerate}

\subsection{Our Results}
As all these techniques fail for real numbers, our challenge in the following is essentially to develop appropriate substitutes. Our core contribution is that we replace the third aforementioned tool: the computation of \emph{sumsets}.

For the integer case, the computation of sumsets (or even more generally, sparse \emph{convolutions}) has a rich history. There are two baseline algorithms: On the one hand, the sumset $A + B$ can naively be computed in time~\makebox{$\Order(|A| \, |B|)$}. On the other hand, if $A, B \subseteq \set{0, \dots, u}$, then the FFT computes the sumset in time~\makebox{$\Order(u \log u)$}. In 1995, Muthukrishnan~\cite{Muthukrishnan95} asked whether these algorithms can be improved, possibly to near-linear time in the output-size $|A + B|$. Cole and Hariharan~\cite{ColeH02} were the first to achieve a Las Vegas algorithm with this running time, leading to a long sequence of follow-up work~\cite{Roche08,MonaganP09,HoevenL12,ArnoldR15,ChanL15,Roche18,Nakos20,GiorgiGC20,BringmannFN21,BringmannFN22,JinX24}. These newer results include a generalization to sparse convolutions of possibly negative vectors~\cite{Nakos20}, optimizations to the conditionally best-possible running time~\cite{BringmannFN21,JinX24} and derandomizations~\cite{ChanL15,BringmannFN22}. 

For real numbers we only have one baseline solution at our disposal, namely the naive $\Order(|A| \, |B|)$-time algorithm. Moreover, all the sophisticated algorithms for integer sumsets unfortunately have in common that they heavily rely on the three problematic tools (scaling, linear hashing and the FFT), and can therefore not be expected to work for real numbers. As our main contribution we develop the first nontrivial algorithm for real sumsets, achieving the near-optimal output-sensitive running time:

\begin{restatable}[Real Sumset]{theorem}{thmsumset} \label{thm:sumset}
There is a Las Vegas algorithm that, given $A, B \subseteq \Real$, computes the sumset $A + B$ in time $\widetilde\Order(|A + B|)$.
\end{restatable}

Perhaps surprisingly, our technique relies on an \emph{algebraic} method. Specifically, we make use of \emph{Prony's method}, an old polynomial interpolation algorithm, in combination with \emph{coprime factorizations} and some combinatorial ideas. In \cref{sec:intro:sec:overview} we provide an overview of our ideas, and in \cref{sec:sumset} we fill in the missing details. This new technique is quite flexible, and we are confident that this combination might have even more applications beyond our scope.

It turns out that the nontrivial computation of sumsets is exactly the primitive that is missing in many algorithms for additive-type problems. Most applications work in a ``plug-and-play'' fashion, while for others (such as for geometric pattern matching) we design new reductions to the sparse sumset problem.

\paragraph{Application 1: Geometric Pattern Matching}
As the first application, consider the geometric pattern matching problem of testing whether, given two point sets $A, B \subseteq \Real^d$, there is a translation $s \in \Real^d$ such that $A + s = \set{a + s : a \in A} \subseteq B$. Thinking of $A$ as a constellation of stars which is to be detected in the night sky $B$, this problem is also called the \emph{Constellation} problem~\cite{CardozeS98,Fischer24}. This is not only a natural problem, but has wide-reaching theoretical and practical applications, for instance in computer vision~\cite{MountNM98,Rucklidge93} and computational chemistry~\cite{FinnKLMSVY97,AkutsuTT98}.

The Constellation problem can naively be solved in time $\Order(|A| \, |B|)$.\footnote{Let $a_0 \in A$ be arbitrary. Then any feasible shift $s$ must stem from the set $B - a_0$. For each of these $|B|$ candidates~$s$ we can test in time $\Order(|A|)$ whether the shift satisfies $A + s \subseteq B$.} In an influential paper, Cardoze and Schulman gave an improved algorithm in time $\widetilde\Order(|A| + |B|)$ for the \emph{integer} Constellation problem, and derived many applications for related (often approximate) variants of point pattern matching. This algorithm was later turned into a Las Vegas algorithm~\cite{ColeH02} and only recently derandomized~\cite{ChanL15,Fischer24}.

All this time, the real Constellation problem remained open. The fastest known algorithm runs in quadratic time, and whether a faster algorithm exists was explicitly asked in~\cite{BarequetH01}.\footnote{In~\cite{BarequetH01}, Barequet and Har-Peled further attribute the question to Kosaraju. In their context the Constellation problem is a special case of the more general 3SUM-hard problem of testing whether a set of intervals can be translated to be covered by another set of intervals.} Our contribution is an affirmative answer to this question, and we provide a near-linear-time (and thus near-optimal) algorithm:

\begin{restatable}[Real Constellation]{theorem}{thmconstellation} \label{thm:constellation}
Let $A, B \subseteq \Real$. There is a Las Vegas algorithm computing $S = \set{ s \in \Real : A + s \subseteq B }$ in time~\smash{$\widetilde\Order(|A| + |B|)$}.
\end{restatable}

To prove \cref{thm:constellation} we design a new reduction from the Constellation problem to computing sparse sumsets. This reduction is inspired by Fischer's recent derandomization~\cite{Fischer24}, which relies on a combination of the Baur-Strassen theorem and a scaling trick. 

We remark that \cref{thm:constellation} readily generalizes to more realistic point pattern matching problems. For instance, as a simple corollary the same theorem holds for $d$-dimensional points with an overhead of $\poly(d)$ in the running time. Additionally, in time $\widetilde\Order((|A| + |B|)^d)$ we can test whether a set of $d$-dimensional points $A$ can be translated and \emph{rotated} to be covered by $B$. For more details on these corollaries, we refer the reader to the discussions in~\cite{CardozeS98,Fischer24}.

\paragraph{Application 2: 3SUM with Preprocessing}
The famous 3SUM problem is to test whether three given sets $A, B, C$ contain a solution to $a + b = c$, $a \in A, b \in B, c \in C$. This problem can be solved in quadratic time, and the conjecture that no substantially faster algorithm exists has evolved into one of the three main pillars of fine-grained complexity. Nevertheless, investigating a problem first studied by Bansal and Williams~\cite{BansalW12}, in an innovative paper Chan and Lewenstein~\cite{ChanL15} discovered that any integer 3SUM instance $(A', B', C')$ can be solved in truly subquadratic time~$\widetilde\Order(n^{13/7})$ if we are allowed to \emph{preprocess} some size-$n$ supersets~\makebox{$A \supseteq A', B \supseteq B', C \supseteq C'$} in quadratic time. Their approach is based on an algorithmic version of the celebrated Balog-Szemerédi-Gowers (BSG) theorem~\cite{BalogS94,Gowers01} from additive combinatorics, in combination with sparse sumset algorithms.

While the running time was later improved~\cite{ChanWX23,KasliwalPS25},\footnote{Chan, Vassilevska Williams and Xu~\cite{ChanWX23} proved that the query time can be reduced to $\widetilde\Order(n^{11/6})$ for integer 3SUM. In their algorithm they rely on the equivalence of integer 3SUM and integer Convolution-3SUM, which itself relies on linear hashing and therefore fails in our context. In even more recent work (announced after this paper was completed), Kasliwal, Polak and Sharma~\cite{KasliwalPS25} have reduced the running time to~\smash{$\widetilde\Order(n^{3/2})$} based purely on linear hashing and the FFT, and without relying on a BSG-type theorem.} it remained open whether a similar algorithm exists for real 3SUM. The BSG theorem works equally well for the reals (in fact, for any Abelian group); the only missing ingredient is an efficient algorithm for computing sparse sumsets (in~\cite[Remark~6.3]{ChanL15}, Chan and Lewenstein also comment that their algorithm can be adapted to the reals in an unconventional computational model which requires flooring). Equipped with \cref{thm:sumset} we can replicate the full result in the standard real RAM model; see \cref{sec:3sum} for details.

\begin{restatable}[Real 3SUM with Preprocessing]{theorem}{thmthreesumpreprocessing} \label{thm:3sum-preprocessing}
Let $A, B, C \subseteq \Real$ be sets of size at most $n$. In Las Vegas time $\widetilde\Order(n^2)$ we can preprocess $A, B, C$ into a data structure with space $\widetilde\Order(n^{13/7})$, such that upon query of sets $A' \subseteq A, B' \subseteq B, C' \subseteq C$, we can solve the 3SUM instance $(A', B', C')$ in Las Vegas time $\widetilde\Order(n^{13/7})$.
\end{restatable}

\paragraph{Application 3: Subset Sum}
As our third and final application we consider the classical Subset Sum problem: Given a (multi-)set $X$ of $n$ nonnegative numbers and a target $t$, determine whether the set of subset sums $\mathcal S(X) := \set{\sum_{x \in X'} x : X' \subseteq X}$ contains the target $t$. Subset Sum is arguably among the simplest NP-hard problems and constitutes the core of many other optimizations problems like Knapsack and Integer Programming. It can be solved in time $2^{n/2}$ (up to polynomial factors) using the textbook meet-in-the-middle algorithm, and it remains a major open problem to improve this running time. In the integer case another line of research that recently gained traction studies \emph{pseudo-polynomial-time} algorithms with respect to several parameters including the target~$t$~\cite{Bellman57,Bringmann17,KoiliarisX19,JinW19}, the maximum item size $\max(X)$~\cite{Pisinger99,PolakRW21,ChenLMZ24} and some others~\cite{BringmannW21}.

Unfortunately, these parameterizations are not reasonable for the real-valued Subset Sum problem: It does not even make sense to think of running times of the form $\widetilde\Order(n + t)$ when~$t$ is an unbounded real number. In the integer case such a running time is considered reasonable as~$t$ naturally bounds the \emph{number} of subset sums in $\mathcal S(X, t) := \mathcal S(X) \cap [0, t]$. In light of this, a reasonable replacement for the parameter $t$ is $|\mathcal S(X, t)|$---i.e., the \emph{output size} for any algorithm computing all subset sums $\mathcal S(X, t)$. Bringmann and Nakos~\cite{BringmannN20} recently initiated the study of this \emph{output-sensitive} Subset Sum problem, with the ultimate goal to obtain an algorithm in time $\widetilde\Order(n + |\mathcal S(X, t)|)$. Their contribution is an elegant algorithm running in time $\widetilde\Order(n + |\mathcal S(X, t)|^{4/3})$, building as their main tool on Ruzsa's triangle inequality from additive combinatorics~\cite{Ruzsa96}. In this work we achieve the same running time for the real variant of Subset Sum:

\begin{restatable}[Real Subset Sum]{theorem}{thmsubsetsumcapped} \label{thm:subsetsum-capped}
There is a Monte Carlo algorithm that, for a multiset $X \subseteq \Realnneg$ and $t \in \Realnneg$, computes $\mathcal S(X, t)$ in time~\smash{$\widetilde\Order(n + |\mathcal S(X, t)|^{4/3})$}.
\end{restatable}

Up to less important technical challenges, we obtain \cref{thm:subsetsum-capped} by redoing Bringmann and Nakos' proof and plugging in our new sumset algorithm. See \cref{sec:sumset-restricted,sec:subsetsum} for details. Along the way we also reproduce, for real numbers, Bringmann and Nakos'~\cite{BringmannN20} results for prefix- and interval-restricted sumsets which we believe to be independently important problems:

\begin{restatable}[Prefix-Restricted Sumset]{theorem}{thmprefixsumset} \label{thm:prefix-sumset}
There is a Las Vegas algorithm that, given $A, B \subseteq \Real$ and $u \in \Real$, computes $C = (A + B) \cap (-\infty, u]$ in time $\widetilde\Order(|A| + |B| + |C|^{4/3})$.
\end{restatable}

\begin{restatable}[Interval-Restricted Sumset]{theorem}{thmintervalsumset} \label{thm:interval-sumset}
There is a Las Vegas algorithm that, given $A, B \subseteq \Real$ and $\ell, u \in \Real$, computes $C = (A + B) \cap [\ell, u]$ in time $\widetilde\Order(|A| + |B| + \sqrt{|A| \, |B| \, |C|})$.
\end{restatable}

\subsection{The Real RAM Model and its Downsides}
All of our algorithms assume the standard real RAM model supporting basic arithmetic operations (${+}, {-}, {\cdot}, {/}$) and comparisons (${=}, {\leq}, {\geq}, {<}, {>}$) of \emph{infinite-precision} real numbers. In particular, we do \emph{not} presuppose stronger operations such as computing exponential or logarithm functions, or flooring (see e.g.~\cite{EvdHM20} or \cref{sec:preliminaries} for more details). In this sense, our algorithms qualify as satisfying \emph{theoretical} solutions to the problems described before.

Nevertheless, the real RAM model often reflects the real world quite inaccurately, and also in our case we consider it unlikely that our algorithms have immediate \emph{practical} impact. The reason is that our algorithms crucially rely on an \emph{exact, infinite-precision} model. This can be realistic in rare cases, e.g.\ when a fast library for rational numbers is available (our algorithms work just as well over the rationals, assuming that arithmetic operations take unit time). But in the majority of cases this assumption is simply unrealistic. We leave it as future work to come up with algorithms in more realistic models, perhaps along the following lines:
\begin{itemize}
    \setlength\parskip{0pt}
    \item A reasonable weaker model is the real RAM model restricted to \emph{bounded-degree comparisons} as suggested in~\cite{ChanWX22}. This could concretely mean that we are only allowed to compare two inputs ($a \leq b$), or compare sums of inputs ($a + b \leq c$), or testing whether a bounded-degree polynomial of the inputs is nonnegative.
    
    Unfortunately, we immediately face a barrier when attempting to compute sumsets in this restricted model. It is a folklore assumption (related to the combinatorial Boolean matrix multiplication conjecture~\cite{VassilevskaWilliams18}) that sumsets (or equivalently, Boolean convolutions) cannot be computed in truly subquadratic time without using algebraic methods. Under this assumption, in a model where we restrict the power of multiplication any efficient sumset algorithm would have to first ``identify'' that the input consists of integers, then convert the input to integers and apply the FFT on the integer registers of the machine. This appears quite strange, and might indicate that the strong version of the real RAM is indeed necessary.
    
    \item Another avenue is to explore algorithms that behave numerically stable. But it is unclear what this should even mean in our context, as already the definition of many our problems assumes exact arithmetic. Besides, for many approximate variants (such as if we tolerate additive $\pm \epsilon$ errors) a practical approach would be to treat the instance as an integer instance obtained by rounding (see e.g.~\cite{CardozeS98}), in which case the integer toolkit becomes available and we have stripped the problem of its real hardness.
\end{itemize}
Having clarified these restrictions, we emphasize that when restricted to \emph{integer} inputs all of our algorithms can be implemented in the word RAM model in essentially the same running time (with only small modifications such as computing modulo an appropriate prime). We thereby also introduce an alternative technique for integer problems, which we believe has broader applicability beyond the problems discussed here.

\subsection{Technical Overview} \label{sec:intro:sec:overview}
In this section we will summarize the key ideas leading to our new algorithm for computing real sumsets (\cref{thm:sumset}). We will first sketch a new approach for integer sets, and will then later deal with the additional challenges posed by real numbers.

\subsubsection{An Algebraic Algorithm for Integer Sumsets}
As a first step, let us reformulate the problem of computing the sumset $A + B$ in terms of polynomials. Specifically, let $f(X) = \sum_{a \in A} X^a$ and $g(X) = \sum_{b \in B} X^b$, and denote their product polynomial by $h := f g$. It is easy to check that $A + B$ is exactly the \emph{support} $\supp(h)$ (i.e., the set of all~$c$ for which the monomial~$X^c$ has a nonzero coefficient). Let~\makebox{$t = |A + B|$}.\footnote{Throughout this overview we will implicitly assume that $t$ is known; this assumption can be removed, e.g.~by exponential search.} Note that all three polynomials $f$, $g$ and $h$ are \emph{$t$-sparse} (i.e., have at most~$t$ nonzero coefficients). In order to compute $h$ we follow an evaluation--interpolation approach:
\begin{enumerate}[itemsep=\smallskipamount]
    \item Evaluate $f(x_0), \dots, f(x_{2t-1})$ and $g(x_0), \dots, g(x_{2t-1})$ at some strategic points $x_0, \dots, x_{2t-1}$.
    \item Multiply $h(x_0) \gets f(x_0) \cdot g(x_0), \dots, h(x_{2t-1}) \gets f(x_{2t-1}) \cdot g(x_{2t-1})$.
    \item Interpolate $h$ from the evaluations $h(x_0), \dots, h(x_{2t-1})$.
\end{enumerate}

\paragraph{Sparse Evaluation and Interpolation}
How fast can we implement this algorithm? The second step runs in linear time $\Order(t)$. And while perhaps not a priori clear, the first step can also be efficiently implemented in time $\widetilde\Order(t)$ with comparably little effort if we choose the evaluation points to be a geometric progression:

\begin{lemma}[Sparse Evaluation and Interpolation] \label{lem:sparse-eval}
Let $f$ be a $t$-sparse polynomial over a field~$\Field$, and let $\omega \in \Field$ have multiplicative order at least $\deg(f)$.\footnote{The condition that~$\omega$ requires large multiplicative order is technical and we encourage the reader to not worry about this condition throughout.} Then:
\begin{itemize}
    \item Given $f$ we can compute $f(\omega^0), \dots, f(\omega^{t-1})$ in time $\Order(t \log^2 t)$.
    \item Given $f(\omega^0), \dots, f(\omega^{t-1})$ and $\supp(f)$, we can uniquely interpolate $f$ in time $\Order(t \log^2 t)$.
\end{itemize}
\end{lemma}

This lemma is an immediate consequence of the fact that linear algebra operations (such as matrix-vector products, solving linear systems, computing determinants, etc.) involving transposed Vandermonde matrices can be implemented in time $\Order(t \log^2 t)$~\cite{KaltofenL88,CannyKY89,Li00,Pan01}. The second item in the lemma also helps in implementing step 3 if we can compute $\supp(h)$---and as we will see next, computing $\supp(h)$ turns out to be the hardest part of the algorithm.

We remark that this general recipe is not new and has already been used for computing sparse sumsets~\cite{ArnoldR15,BringmannFN22}; however, both of these works compute $\supp(h)$ by a scaling trick which cannot work for real numbers and is thus prohibitive for us. We will instead employ \emph{Prony's method}, which we describe in some detail in the next subsection.

\paragraph{Prony's Method}
\emph{Prony's method}, named after its inventor Gaspard de Prony, is an old polynomial interpolation algorithm dating back to 1795~\cite{Prony1795}. It was rediscovered several times since then, for decoding BCH codes~\cite{Wolf67} and in the context of polynomial interpolation~\cite{BenOrT88}; see also~\cite{Roche18}. It has since found several surprising applications in algorithm design, see e.g.~\cite{CliffordEPR09,CliffordKP19} for two applications in string algorithms.

As we do not expect the target audience of this paper to be necessarily familiar with Prony's method, we include a detailed overview. In this section our aim is to be more intuitive than formally precise. Prony's method states that we can recover the support of a $t$-sparse polynomial $h$ from evaluations at a geometric progression of length $2t$,~\makebox{$h(\omega^0), \dots, h(\omega^{2t-1})$}. We start with a definition:

\begin{definition}[Linear Recurrence] \label{def:linear-recurrence}
A sequence~$s_0, s_1, \dots, s_{t-1}$ is \emph{linearly recurrent with degree~$r$} if there is a degree-$r$ polynomial $\Lambda(X) = \sum_{\ell=0}^r \lambda_\ell X^\ell$ such that each term in the sequence is determined by a linear combination of its $r$ preceding terms, weighted with $\lambda_0, \dots, \lambda_r$:
\begin{equation*}
    \sum_{\ell=0}^r \lambda_\ell \cdot s_{i+\ell} = 0 \quad\text{for all $0 \leq i < t - r$.}
\end{equation*}
If $\Lambda$ is the monic polynomial (i.e., with leading coefficient $\lambda_r = 1$) with smallest-possible degree satisfying this condition then we call $\Lambda$ the \emph{minimal polynomial} of the sequence $s_0, s_1, \dots, s_{t-1}$.
\end{definition}

At first it might seem odd why we view the $\lambda_i$'s as a polynomial, but this viewpoint hopefully soon starts to make sense. It turns out that the sequence of evaluations $h(\omega^0), \dots, h(\omega^{2t-1})$ is linearly recurrent with degree exactly~$t$. And, more importantly, we can read of the support from the minimal polynomial of this recurrence as stated in the following lemma. This lemma constitutes the heart of Prony's method:

\begin{lemma}[Prony's Method] \label{lem:prony}
Let $h$ be a polynomial which is exactly $t$-sparse, and let $\omega$ have multiplicative order at least $\deg(h)$. Then the sequence $h(\omega^0), \dots, h(\omega^{2t-1})$ is linearly recurrent with degree $t$. Moreover, its minimal polynomial is
\begin{equation*}
    \Lambda(X) = \prod_{c \in \supp(h)} (X - \omega^c).
\end{equation*}
\end{lemma}
\begin{proof}
We include a proof at this point as we believe it to be quite informative. Let $c_0, \dots, c_{t-1}$ denote the support of $h$ and let \smash{$h(X) = \sum_{j=0}^{t-1} h_j \cdot X^{c_j}$}. Assume that the sequence $h(\omega^0), \dots, h(\omega^{2t-1})$ is linearly recurrent with degree $r$ for some polynomial $\Lambda(X) = \sum_{\ell=0}^r \lambda_\ell X^\ell$. We can therefore express the recurrence condition as follows, for any $i \geq 0$:
\begin{gather*}
    0 = \sum_{\ell=0}^r \lambda_\ell \cdot h(\omega^{i+\ell}) \\
    \qquad= \sum_{\ell=0}^r \lambda_\ell \sum_{j=0}^{t-1} h_j \cdot \omega^{(i+\ell)c_j} \\
    \qquad= \sum_{j=0}^{t-1} \omega^{i c_j} \cdot h_j \cdot \sum_{\ell=0}^r \lambda_\ell \cdot \omega^{\ell c_j} \\
    \qquad= \sum_{j=0}^{t-1} \omega^{i c_j} \cdot h_j \cdot \Lambda(\omega^{c_j}).
\end{gather*}
We see two consequences: First, picking the polynomial $\Lambda(X) = \prod_{j=0}^{t-1} (X - \omega^{c_j})$ with roots at all powers $\omega^{c_0}, \dots, \omega^{c_{t-1}}$ satisfies this recurrence condition for all $i$. In particular, the recurrence has degree at most $t$.

Second, suppose that the minimal polynomial $\Lambda$ has degree less than $t$. Then for at least one power $\omega^{c_0}, \dots, \omega^{c_{t-1}}$, $\Lambda$ does not have a root. Let~\smash{$p(X) = \sum_{j=0}^{t-1} h_j \cdot \Lambda(\omega^{c_j}) \cdot X^{c_j}$}, and note that~$p$ is not identically zero. We can rewrite the recurrence from before as $p(\omega^i) = 0$ for all $i \in [t]$. This yields a contradiction: On the one hand $p$ is not identically zero, but on the other hand, $p$ is $t$-sparse and therefore uniquely determined by its evaluations $p(\omega^0), \dots, p(\omega^{t-1})$ all of which are zero (see \cref{lem:sparse-eval}).
\end{proof}

Moreover, it is well-know that we can compute the degree of a linear recurrence as well as its minimal polynomial. This can be achieved in quadratic-time by the well-known Berlekamp-Massey algorithm~\cite{Berlekamp68,Massey69}, but also in near-linear time by an application of the Extended Euclidian Algorithm~\cite{KaltofenL88,vonzurGathenG13}:

\begin{lemma}[Solving Linear Recurrences] \label{lem:linear-recurrence}
Let $s_0, \dots, s_{t-1} \in \Field$ for some field $\Field$. The minimal polynomial of this sequence can be computed in time $\Order(t \log^2 t)$.
\end{lemma}

\paragraph{Putting the Pieces Together}
\cref{lem:prony} suggests the following three steps to fill in the missing implementation of step 3:
\begin{enumerate}[itemsep=\smallskipamount]
    \item Evaluate $f(\omega^0), \dots, f(\omega^{2t-1})$ and $g(\omega^0), \dots, g(\omega^{2t-1})$ (\cref{lem:sparse-eval}).
    \item Multiply $h(\omega^0) \gets f(\omega^0) \cdot g(\omega^0), \dots, h(\omega^{2t-1}) \gets f(\omega^{2t-1}) \cdot g(\omega^{2t-1})$.
    \item[3a.] Compute the minimal polynomial $\Lambda$ of the sequence $h(\omega^0), \dots, h(\omega^{2t-1})$ (\cref{lem:linear-recurrence}).
    \item[3b.] Compute the roots $\omega^{c_0}, \dots, \omega^{c_{t-1}}$ of $\Lambda$.
    \item[3c.] Compute the logarithms $c_0, \dots, c_{t-1}$ of the roots to the base $\omega$.
\end{enumerate}

It turns out that we can indeed carry the computations of Steps~3b and~3c efficiently over some appropriate chosen \emph{finite field}; see e.g.~\cite{Kaltofen10}. This involves several number-theoretic tricks which, inconveniently, cannot work for real numbers. Interestingly, with the new ideas we are about to present, even these tricks will not be necessary anymore.

\subsubsection{Adapting the Algorithm to Real Sumsets}
We emphasize that up until this point we have not invented anything new, but rather viewed the existing techniques in a different light. So let us finally get back to the original problem of computing sumsets of \emph{real} numbers. Is it conceivable that this approach succeeds for real numbers? Let~\makebox{$A, B \subseteq \Real$}. We similarly define functions $f, g, h : \Real \to \Real$ by $f(X) = \sum_{a \in A} X^a$, $g(X) = \sum_{b \in B} X^b$ and $h(X) = f(X) \cdot g(X)$ (though these are no longer polynomials). Apart from this difference in naming it can be checked that all steps are well-defined and reasonable. In particular, the sequence of evaluations $h(\omega^0), \dots, h(\omega^{2t-1})$ (for some real number $\omega \in \Real$) is still linearly recurrent with degree exactly $t = |A + B|$ and its minimal polynomial is $\Lambda(X) = \prod_{c \in A + B} (X - \omega^c)$. Thus, the same algorithm as before is correctly computing the set $A + B$.

Unfortunately, there are two major issues with this algorithm. First, Step~3b requires the computation of the real roots of the degree-$t$ polynomial $\Lambda$. As discussed before, this step could only be made efficient by working over an appropriate finite field, and these tricks are fruitless for the reals. The second problem seems even more hopeless: The real RAM model simply does not support computing the exponentiations $\omega^0, \dots, \omega^{2t-1}$ in Step~1, and taking the logarithms in Step~3c. In the following two paragraphs we present fixes for these two problems.

\paragraph{Solution 1: Coprime Factorizations}
While there are many specialized algorithms for computing the roots of polynomials over finite fields, the integers, or the rationals (see e.g.~\cite{vonzurGathenG13}), to our knowledge there is no general-purpose near-linear-time root finding algorithm for abstract fields. Another line of research has focused on \emph{approximately} computing the real roots based on numerical methods (but of course even the smallest error is intolerable in our setup). All in all, we are not aware of any algorithm that exactly computes the real roots of a polynomial in near-linear time, leaving us in an unclear situation.

Our solution is to settle for a weaker notion of factorization---a so-called \emph{coprime factorization}. The background is that factoring a polynomial or an integer into its irreducible factors (i.e., primes) is typically a hard problem. In contrast, it is a much simpler problem to factor several polynomials or integers into factors that are only required to be \emph{pairwise coprime}.\footnote{In this setup a single polynomial or integer can be trivially be factored into the singleton set containing itself. In this set all factors are vacously pairwise coprime.} Bernstein~\cite{Bernstein05} gave a near-linear-time algorithm to compute coprime factorizations over arbitrary abstract fields:

\begin{restatable}[Coprime Factorization of Polynomials~\cite{Bernstein05}]{theorem}{thmcoprimebasis} \label{thm:coprime-basis}
Given a set $\mathcal P$ of nonzero polynomials (over some field $\Field$), in time $\widetilde\Order(\sum_{P \in \mathcal P} \deg(P))$ we can compute a set of polynomials $\mathcal Q$ such that:
\begin{itemize}
    \item Each polynomial $P \in \mathcal P$ can be expressed as the product of polynomials in $\mathcal Q$ (up to scalars), and each polynomial $Q \in \mathcal Q$ appears as a factor of some polynomial $P \in \mathcal P$.
    \item The polynomials in $\mathcal Q$ are pairwise coprime (i.e., $\gcd(Q, Q') = 1$ for all distinct $Q, Q' \in \mathcal Q$).
    \item \smash{$\sum_{Q \in \mathcal Q} \deg(Q) \leq \sum_{P \in \mathcal P} \deg(P)$}.
\end{itemize}
\end{restatable}

In fact, Bernstein presents this result in a broader algebraic context that not only applies to polynomials. Since Bernstein employs a nonstandard model of computation, we include a self-contained proof of \cref{thm:coprime-basis} in \cref{sec:coprime}. We consider this a very interesting tool which, to our surprise, has only limited applications so far.

To apply this theorem we depart from the algebraic framework and attempt a more combinatorial solution. In summary, the previously outlined sumset algorithm allows us to compute in near-linear time a polynomial $\Lambda = \Lambda_{A, B}$ defined by
\begin{equation*}
    \Lambda_{A, B}(X) = \prod_{c \in A + B} (X - \omega^c).
\end{equation*}
Instead of factoring this polynomial directly, our idea is to repeat this algorithm for a polylogarithmic number of times with sets $A_1, \dots, A_m \subseteq A$ and $B_1, \dots, B_m \subseteq B$ to compute the set of polynomials $\mathcal P := \set{\Lambda_{A_1, B_1}, \dots, \Lambda_{A_m, B_m}}$. On these polynomials we apply \cref{thm:coprime-basis} to compute a coprime factorization. Let us call two factors~\makebox{$X - \omega^{c_1}$} and~\makebox{$X - \omega^{c_2}$} \emph{separated} if there some polynomial $\Lambda_{A_i, B_i}$ that contains~\makebox{$X - \omega^{c_1}$} as a factor but does not contain~\makebox{$X - \omega^{c_2}$} as a factor. The key insight is that any separated pair $X - \omega^{c_1}$ and~\makebox{$X - \omega^{c_2}$} does not appear in the same factor in $\mathcal Q$. Indeed, suppose that both $X - \omega^{c_1}$ and~\makebox{$X - \omega^{c_2}$} appear in the same factor~\makebox{$Q \in \mathcal Q$}. By the coprime condition this is the only factor $Q$ in which~\makebox{$X - \omega^{c_1}$} appears. Therefore, $Q$ must be a factor of~$\Lambda_{A_i, B_i}$. But this yields a contradiction as we assumed that $X - \omega^{c_2}$ is not a factor of~$\Lambda_{A_i, B_i}$. In light of this insight, in order to obtain the complete factorization~\makebox{$\mathcal Q = \set{X - \omega^c : c \in A + B}$}, it suffices to guarantee that all pairs of factors $X - \omega^{c_1}$, $X - \omega^{c_2}$ are separated. Luckily, this can be achieved by the following lemma:

\begin{restatable}[Random Restrictions]{lemma}{lemrandomrestrictions} \label{lem:random-restrictions}
There is a Monte Carlo algorithm that, given sets $A, B \subseteq \Real$ of size at most $n$, computes in time $\Order(n m)$ subsets $A_1, \dots, A_m \subseteq A,\, B_1, \dots, B_m \subseteq B$ such that:
\begin{itemize}
    \item For every pair $c_1, c_2 \in A + B$, there is some $i \in [m]$ with $\abs{(A_i + B_i) \cap \set{c_1, c_2}} = 1$ (with high probability).
    \item $m = \Order(\log^2 n)$.
\end{itemize}
\end{restatable}

The insight behind \cref{lem:random-restrictions} is a rather simple subsampling procedure: By subsampling $A$ with rates $2^{-0}, 2^{-1}, \dots, 2^{-\log |A|}$ we obtain with good probability a subset $A'$ such that there is a unique witness $a \in A'$ of $\set{c_1, c_2}$; that is, there exist $b_1 \in B$ such that $a + b_1 = c_1$ or there exists $b_2 \in B$ such that $a + b_2 = c_2$. If only one of these events happens, then we have successfully separated $c_1$ and $c_2$. Otherwise, we subsample $B$ with rate $\frac12$ to eliminate exactly one of the two events with constant probability. For more details see \cref{sec:sumset:sec:random-restrictions}.

\paragraph{Solution 2: Prony's Method for Other Operators}
Recall that the second challenge is to eliminate the need to compute exponentiations and logarithms (as these operations are not supported by the real RAM model). It turns out that Prony's interpolation method works for operators other than ``polynomial evaluation''. That is, we can interpolate a polynomial $f$ not only given the evaluations $f(\omega^0), \dots, f(\omega^{t-1})$ (which is perhaps the most natural representation), but also from evaluations of its derivatives \smash{$\frac{d^0 f}{dx^0}(1), \dots, \frac{d^{2t-1} f}{dx^{2t-1}}(1)$}~\cite{Huang23}, or even from evaluations of more general operators~\cite{StampferP20}.

With these works in mind, let us rework some of the previous setup to better match the real RAM model. Let $A, B \subseteq \Real$, let $f, g : \Real \to \Real$ denote the indicator functions of $A$ and $B$ (satisfying that $f(a) = 1$ if $a \in A$ and $f(a) = 0$ otherwise). We write $\supp(f) = \set{x \in \Real : f(x) \neq 0}$. Let~$h$ denote the discrete convolution of $f$ and $g$, defined by $h(c) = \sum_{a + b = c} f(a) \cdot g(b)$ (which is well-defined since $f, g$ have only finite support). For integers $i \geq 0$, we define the \emph{sum operator} $\Sigma^i (f)$ as the quantity
\begin{equation*}
    \Sigma^i(f) = \sum_{x \in \Real} x^i \cdot f(x).
\end{equation*}
Note that $\Sigma^i(f)$ can indeed be computed in the real RAM model as all powers $i$ are integers and can be computed efficiently by repeated squaring. Moreover, analogously to before it turns out that for any function $h : \Real \to \Real$ with support size $t$, the sequence $\Sigma^0(h), \dots, \Sigma^{2t-1}(h)$ is linearly recurrent with degree $t$, and its minimal polynomial is
\begin{equation*}
    \Lambda(X) = \prod_{a \in \supp(f)} (X - a).
\end{equation*}
This suggests the following algorithm: We evaluate $\Sigma^0(f), \dots, \Sigma^{2t-1}(f)$ and $\Sigma^0(g), \dots, \Sigma^{2t-1}(g)$, somehow multiply these evaluations to obtain $\Sigma^0(h), \dots, \Sigma^{2t-1}(h)$, and then apply the modified Prony's method to compute and factor $\Lambda(X) = \prod_{a \in \supp(f)} (X - a)$. This alternative setup, while somewhat less intuitive, indeed avoids exponentiations and logarithms. It comes, however, with additional challenges such as how to efficiently evaluate $\Sigma^0(f), \dots, \Sigma^{2t-1}(f)$, and how to combine these evaluations to efficiently obtain $\Sigma^0(h), \dots, \Sigma^{2t-1}(h)$. All of these challenges have natural algebraic solutions which we provide in the technical \cref{sec:sumset}.
\section{Preliminaries} \label{sec:preliminaries}
We write $[n] = \set{0, \dots, n-1}$ and typically index objects starting from $0$.

\paragraph{Machine Model}
Throughout we work in the standard real RAM model. In this model we assume the standard word RAM model with word size $\Theta(\log n)$ (where $n$ is the input size of the problem) and equipped with special machine cells holding real numbers. The model supports basic arithmetic operations (${+}, {-}, {\cdot}, {/}$) and comparisons (${=}, {\leq}, {\geq}, {<}, {>}$) on reals in constant time. Moreover, we can convert in constant time any integer (stored in a word cell) into a real (stored in a real cell). We refer to~\cite{EvdHM20} for a more formal definition.

We frequently deal with sets of real numbers. As we cannot use hash maps in this setting (as there is no way to hash from reals to integer indices), we implicitly assume that these sets are stored either as sorted arrays or balanced binary trees. The log-factors incurred to the running times will never matter for us.

\paragraph{Randomization}
We say that an event occurs with high probability if it occurs with probability at least $1 - n^{-c}$, where $n$ is the problem size and $c$ is an arbitrarily chosen constant. We say that a problem is in Monte Carlo time~$T$ if it can be solved by a Monte Carlo algorithm running in time~$T$ and succeeding with high probability. Analogously, we say that it is in Las Vegas time~$T$ if there is a Las Vegas (i.e., zero-error) algorithm running in expected time $T$, and terminating in time $T$ with high probability.
\section{Real Sumset in Near-Linear Time} \label{sec:sumset}
In this section we design a near-linear time algorithm to compute the sumset $A + B$ for two sets of real numbers $A, B \subseteq \Real$. In fact, we solve the more general problem of computing the discrete convolution of two real functions. Specifically, throughout this section we consider functions~\makebox{$f, g : \Real \to \Real$} \emph{with finite support} (i.e., for which $\supp(f) = \set{x \in \Real : f(x) \neq 0}$ is a finite set). Computationally, we represent functions $f$ with finite support as a finite list of its nonzero entries~\makebox{$(x, f(x))$}.

We define the \emph{discrete convolution} $f \conv g : \Real \to \Real$ as the function defined by
\begin{equation*}
    (f \conv g)(z) = \sum_{\substack{x, y \in \Real\\x + y = z}} f(x) \cdot g(y).
\end{equation*}
Note that $f \conv g$ again has finite support. Convolutions are tightly connected to sumsets in the following way: For a finite set $A \subseteq \Real$, let $\One_A : \Real \to \Real$ denote the indicator function of $A$ (which is $1$ on $A$ and $0$ elsewhere). Then $A + B = \supp(\One_A \conv \One_B)$, as can easily be verified from the definition.

We structure this section as follows: In \cref{sec:sumset:sec:prony} we tailor Prony's method to our need in the context of real numbers up to some missing pieces. In \cref{sec:sumset:sec:sumset-size} we demonstrate how to compute the support size of the convolution, and in \cref{sec:sumset:sec:random-restrictions} we show how to replace the polynomial factorization step.

\subsection{A Modified Prony's Method} \label{sec:sumset:sec:prony}
We start with the definition of the ``sum operator'', and show that it satisfies a ``product rule''.
\begin{definition}[Sum Operator] \label{def:op}
Let $f : \Real \to \Real$ (with finite support) and $i \in \Nat$. We define
\begin{equation*}
    \Sigma^i(f) = \sum_{x \in \Real} x^i \cdot f(x).
\end{equation*}
\end{definition}

\begin{lemma}[Product Rule] \label{lem:op-product-rule}
Let $f, g \in \Real \to \Real$ (with finite support). Then, for any $k \in \Nat$ we have that
\begin{equation*}
    \Sigma^k(f \conv g) = \sum_{\substack{i, j \in \Nat\\i+j=k}} \binom{k}{i} \cdot \Sigma^i(f) \cdot \Sigma^j(g).
\end{equation*}
\end{lemma}
\begin{proof}
The proof is a straightforward calculation involving the binomial theorem:
\begin{align*}
    \Sigma^k(f \conv g) &= \sum_{z \in \Real} z^k \cdot (f \conv g)(z) \\
    &= \sum_{z \in \Real} z^k \sum_{\substack{x, y \in \Real\\x + y = z}} f(x) \cdot g(y) \\
    &= \sum_{x, y \in \Real} (x + y)^k \cdot f(x) \cdot g(y) \\
    &= \sum_{\substack{i, j \in \Nat\\i + j = k}} \binom{k}{i} \cdot \sum_{x, y \in \Real} x^i \cdot f(x) \cdot y^j \cdot g(y) \\
    &= \sum_{\substack{i, j \in \Nat\\i + j = k}} \binom{k}{i} \cdot \Sigma^i(f) \cdot \Sigma^j(g). \qedhere
\end{align*}
\end{proof}

Next, we prove that, for any function $f$, we can quickly evaluate the sum operator and interpolate $f$ from its evaluations (\cref{lem:eval-op}) and that we can maintain the evaluations of the sum operator under taking convolutions (\cref{lem:conv-op}). 

\begin{lemma}[Evaluation and Interpolation] \label{lem:eval-op}
Let $f : \Real \to \Real$ be $t$-sparse. Then there are deterministic algorithms in time $\Order(t \log^2 t)$ for the following two tasks:
\begin{itemize}
    \item Given $f$, we can evaluate $\Sigma^0(f), \dots, \Sigma^{t-1}(f)$.
    \item Given $\Sigma^0(f), \dots, \Sigma^{t-1}(f)$ and $\supp(f)$, we can (uniquely) interpolate $f$.
\end{itemize}
\end{lemma}
\begin{proof}
Let $a_0, \dots, a_{t-1}$ denote the support of $f$, and consider the following obvious identity:
\begin{equation*}
	\begin{bmatrix}
		\Sigma^0(f) \\
		\Sigma^1(f) \\
		\vdots \\
		\Sigma^{t-1}(f)
	\end{bmatrix}
	=
	\begin{bmatrix}
		1 & 1 & \cdots & 1 \\
		a_0 & a_1 & \cdots & a_{t-1} \\
		\vdots & \vdots & \ddots & \vdots \\
		a_0^{t-1} & a_1^{t-1} & \cdots & a_{t-1}^{t-1}
	\end{bmatrix}
	\begin{bmatrix}
		f(a_0) \\
		f(a_1) \\
		\vdots \\
		f(a_{t-1})
	\end{bmatrix}
\end{equation*}
Both statements follow from the fact that linear algebraic operations involving transposed Vandermonde matrices run in time $\Order(t \log^2 t)$~\cite{KaltofenL88,CannyKY89,Li00,Pan01}. For the first item we compute a matrix-vector product, and for the second item we solve a linear equation system with indeterminates $f(a_0), \dots, f(a_{t-1})$.
\end{proof}

\begin{lemma}[Convolution of Sum Operators] \label{lem:conv-op}
Let $f, g : \Real \to \Real$ (with finite support). Given $\Sigma^0(f), \dots, \Sigma^{t-1}(f), \Sigma^0(g), \dots, \Sigma^{t-1}(g)$, we can compute $\Sigma^0(f \conv g), \dots, \Sigma^{t-1}(f \conv g)$ in time $\Order(t \log t)$.
\end{lemma}
\begin{proof}
The key idea is to use the product rule (\cref{lem:op-product-rule}):
\begin{equation*}
    \Sigma^k(f \conv g) = \sum_{\substack{i, j \in \Nat\\i+j=k}} \binom{k}{i} \cdot \Sigma^i(f) \cdot \Sigma^j(g).
\end{equation*}
Using that $\binom{k}{i} = \frac{k!}{i! \, j!}$, this identity can be rewritten as
\begin{equation*}
    \frac{\Sigma^k(f \conv g)}{k!} = \sum_{\substack{i, j \in \Nat\\i+j=k}} \frac{\Sigma^i(f)}{i!} \cdot \frac{\Sigma^j(g)}{j!}.
\end{equation*}
We can compute the sequences \smash{$(\Sigma^i(f) / i!)_{i=0}^{t-1}$} and \smash{$(\Sigma^j(g) / j!)_{j=0}^{t-1}$} in time $\Order(t)$. Then we view these sequences as the coefficients of two degree-$t$ polynomials over $\Real$; note that~\smash{$(\Sigma^k(f \conv g) / k!)_{k=0}^{t-1}$} corresponds exactly to the first coefficients the product polynomial. Thus, by computing the product of two degree-$t$ polynomials over $\Real$ we can recover the desired values $(\Sigma^k(f \conv g))_{k=0}^{t-1}$. This task can be solved in time $\Order(t \log t)$ by the Fast Fourier Transform (even over general rings, see~\cite{Kaminski88,CantorK91}).
\end{proof}

\begin{lemma}[Prony's Method for the Sum Operator] \label{lem:prony-op}
Let $f : \Real \to \Real$ be $t$-sparse. Then the sequence $\Sigma^0(f), \dots, \Sigma^{2t-1}(f)$ is linearly recurrent with degree exactly $t$, and its minimal polynomial is
\begin{equation*}
    \Lambda(X) = \prod_{a \in \supp(f)} (X - a).
\end{equation*}
\end{lemma}
\begin{proof}
Let $a_1, \dots, a_t$ denote the support of $f$. If the sequence $\Sigma^0(f), \dots, \Sigma^{2t-1}(f)$ is linearly recurrent with some degree $r$ for some polynomial $\Lambda(X) = \sum_{\ell=0}^r \lambda_\ell \cdot X^\ell$, then we can express the recurrence condition as follows, for any $0 \leq i < 2t - r$:
\begin{gather*}
    0 = \sum_{\ell=0}^r \lambda_\ell \cdot \Sigma^{i+\ell}(f) \\
    \qquad= \sum_{\ell=0}^r \lambda_\ell \cdot \sum_{j=1}^t a_j^{i+\ell} \cdot f(a_j) \\
    \qquad= \sum_{j=1}^t a_j^i \cdot f(a_j) \cdot \sum_{\ell=0}^r \lambda_\ell \cdot a_j^\ell \\
    \qquad= \sum_{j=1}^t a_j^i \cdot f(a_j) \cdot \Lambda(a_j).
\end{gather*}
Observe the following two consequences: First, picking the polynomial $\Lambda(X) = \prod_{j=1}^t (X - a_j)$ with roots at all support elements $a_1, \dots, a_t$ satisfies this recurrence condition for all $i$. In particular, the recurrence has degree at most $t$.

Second, suppose that the minimal polynomial $\Lambda$ has degree $r < t$. Then for at least one support element $a_i$, $\Lambda$ does not have a root. Let $g : \Real \to \Real$ be the function supported on $\set{a_1, \dots, a_t}$ defined by~\makebox{$g(a_j) = f(a_j) \cdot \Lambda(a_j)$}. Clearly $g$ is at most $t$-sparse and not identically zero. Moreover, we can rewrite the previous recurrence relation as $\Sigma^i(g) = 0$, for all $0 \leq i < t$. This yields a contradiction: On the one hand $g$ is not identically zero, but on the other hand, $g$ is $t$-sparse and therefore uniquely determined by its evaluations $\Sigma^0(g), \dots, \Sigma^{t-1}(g)$, all of which are zero (see the second item in \cref{lem:eval-op}).
\end{proof}

The following statement combines the previously established lemmas:

\begin{lemma} \label{lem:sumset-minimal-poly}
Let $A, B \subseteq \Real$ be finite sets and let $t \geq |A + B|$. There is a deterministic algorithm running in time $\Order(t \log^2 t)$ computing the polynomial
\begin{equation*}
    \Lambda_{A+B}(X) = \prod_{c \in A + B} (X - c).
\end{equation*}
\end{lemma}
\begin{proof}
Let $f = \One_A$ and $g = \One_B$. Using \cref{lem:eval-op} we evaluate the sequences~\smash{$\Sigma^0(f), \dots, \Sigma^{2t-1}(f)$} and~\smash{$\Sigma^0(g), \dots, \Sigma^{2t-1}(g)$}. Then, using \cref{lem:conv-op} we compute \smash{$\Sigma^0(f \conv g), \dots, \Sigma^{2t-1}(f \conv g)$}. Finally, we solve this linear recurrence (by \cref{lem:linear-recurrence}) and compute its minimal polynomial, which, by \cref{lem:prony-op}, happens to be exactly
\begin{equation*}
    \Lambda_{A + B}(X) = \prod_{c \in \supp(f \conv g)} (X - c) = \prod_{c \in A + B} (X - c).
\end{equation*}
Lemmas~\ref{lem:eval-op},~\ref{lem:conv-op} and~\ref{lem:linear-recurrence} each run in time $\Order(t \log^2 t)$, so the claimed time bound follows. 
\end{proof}

\subsection{Computing the Sumset Size} \label{sec:sumset:sec:sumset-size}
One shortcoming of \cref{lem:sumset-minimal-poly} is that it requires as an input a good approximation of the size $|A + B|$. It is easy to approximate this size by exponential search and a ``Schwartz-Zippel-like evaluation'', but it turns out that in our setting we can even compute $|A + B|$ \emph{exactly} and \emph{deterministically}. This algorithm is based on the following lemma which was first proposed to be useful in this context in~\cite{BringmannFN22}. Since the proof is short and insightful, we include it.

\begin{lemma}[Sparsity Tester for Nonnegative Functions, \cite{Karlin68,Lindsay89}] \label{lem:sparsity-tester}
Let $f : \Real \to \Real_{\geq 0}$ (with finite support). For any $s \geq 0$, $f$ has sparsity more than $s$ if and only if
\begin{equation*}
    \det
    \begin{bmatrix}
        \Sigma^0(f) & \Sigma^1(f) & \cdots & \Sigma^{s-1}(f) \\
        \Sigma^1(f) & \Sigma^2(f) & \cdots & \Sigma^s(f) \\
        \vdots & \vdots & \ddots & \vdots \\
        \Sigma^{s-1}(f) & \Sigma^s(f) & \cdots & \Sigma^{2s-2}(f) \\
    \end{bmatrix}
    > 0.
\end{equation*}
\end{lemma}
\begin{proof}
Let~\smash{$M = [\Sigma^{i+j}(f)]_{i, j \in [s]}$} denote the matrix in the statement. We first prove that $M$ is a positive semi-definite matrix. To this end let $\lambda \in \Real^s$ be arbitrary; we prove that $\lambda^T M \lambda \geq 0$:
\begin{align*}
    \lambda^T M \lambda
    &= \sum_{i, j \in [s]} \lambda_i \cdot \lambda_j \cdot \Sigma^{i+j}(f) \\
    &= \sum_{i, j \in [s]} \lambda_i \cdot \lambda_j \cdot \sum_{x \in \Real} x^{i+j} \cdot f(x) \\
    &= \sum_{x \in \Real} f(x) \cdot \parens*{\sum_{i \in [s]} \lambda_i x^i} \cdot \parens*{\sum_{j \in [s]} \lambda_j x^j} \\
    &= \sum_{x \in \Real} f(x) \cdot \parens*{\sum_{i \in [s]} \lambda_i x^i}^2 \\
    &\geq 0.
\end{align*}
Having established that $M$ is always positive semi-definite, we prove that it is positive definite if and only if $f$ has sparsity more than $s$. For a vector $\lambda \in \Real^s$ as before, let \smash{$\Lambda(X) = \sum_{i \in [s]} \lambda_i X^i$}. The previous calculation reveals that $\lambda^T M \lambda = \sum_{x \in \Real} f(x) \cdot \Lambda(x)^2$. Thus, on the one hand, if $f$ is at most $s$-sparse we can pick $\Lambda$ to be any degree-$s$ polynomial with roots at $\supp(f)$ so that $\lambda^T M \lambda = 0$. On the other hand, if $f$ has sparsity more than $s$ then any nonzero degree-$s$ polynomial $\Lambda$ must be nonzero at some support element $x \in \supp(f)$ and thus $\lambda^T M \lambda > 0$ for all $\lambda \in \Real^s$.
\end{proof}

\begin{lemma}[Computing the Sumset Size] \label{lem:sumset-size}
There is a deterministic algorithm that, given sets \smash{$A, B \subseteq \Real$}, computes the size $t = |A + B|$ in time $\Order(t \log^2 t)$.
\end{lemma}
\begin{proof}
As a first step, we compute a $2$-approximation $s$ of $t$. The idea is to set $s \gets \max\set{|A|, |B|}$ and to repeatedly double $s$ until the condition $\frac s2 \leq |A + B| \leq s$ is satisfied. In each step we are required to test whether $|A + B| \leq s$. To this end, we use the sparsity tester from the previous lemma. Specifically, we let $f = \One_A$ and $g = \One_A$, compute \smash{$\Sigma^0(f), \dots, \Sigma^{2t-1}(f)$} and \smash{$\Sigma^0(g), \dots, \Sigma^{2t-1}(g)$} by \cref{lem:eval-op}, and then compute \smash{$\Sigma^0(f \conv g), \dots, \Sigma^{2t-1}(f \conv g)$} by \cref{lem:conv-op}. We compute the determinant of the Hankel matrix~\smash{$[\Sigma^{i+j}(f \conv g)]_{i, j \in [s]}$} in time $\Order(s \log^2 s)$~\cite{Pan01}. By \cref{lem:sparsity-tester}, this determinant is positive if and only $f \conv g$ has sparsity more than $s$.

Running this procedure for any fixed value of $s$ requires time $\Order(s \log^2 s)$ by Lemmas~\ref{lem:eval-op},~\ref{lem:conv-op} and~\cite{Pan01}. Since we double $s$  in each step until $s = \Theta(t)$, the total running time so far is bounded by $\Order(t \log^2 t)$.

Finally, we apply \cref{lem:sumset-minimal-poly} with $s \geq t$ to compute the polynomial $\prod_{c \in A + B} (X - c)$; we can trivially read off $t$ as the degree of this polynomial. The running time of this final step is also bounded by $\Order(t \log^2 t)$ by \cref{lem:sumset-minimal-poly}. 
\end{proof}

\subsection{Factoring by Random Restrictions} \label{sec:sumset:sec:random-restrictions}
To complete the sumset algorithm it remains to factor the polynomial $\Lambda_{A + B}$ into the linear factors~\makebox{$X - c$}, for $c \in A + B$. At this point however, rather than relying on more algebraic machinery, we switch to a more combinatorial approach. Consider the following lemma:

\lemrandomrestrictions*
\begin{proof}
Let $L = \ceil{\log |A|} + 1$ and consider this random experiment: Sample $\ell \in \set{0, \dots, L - 1}$ uniformly at random. Then take $A'$ to be a uniform subsample of $A$ with rate $2^{-\ell}$, and let $B'$ be a uniform subsample of $B$ with rate $\frac12$. We claim that the pair $(A', B')$ satisfies the lemma statement with good probability:

\begin{claim} \label{lem:random-restrictions:clm:single}
For any $c_1, c_2 \in A + B$, we have $|(A' + B') \cap \set{c_1, c_2}| = 1$ with probability at least~$\frac{1}{16L}$.
\end{claim}

Suppose for a moment that the claim is correct. Then we can easily derive the lemma by letting each pair $(A_i, B_i)$ in the output be sampled as in the previous paragraph. By picking $m = 160 L \log n = \Order(\log^2 n)$, each pair $c_1, c_2$ fails with probability at most
\begin{equation*}
	\parens*{1 - \frac{1}{16 L}}^m \leq \exp\parens*{-\frac{m}{16L}} \leq n^{-10}.
\end{equation*}
Taking a union bound over the at most $|A + B|^2 \leq n^4$ choices of $c_1$ and $c_2$, the lemma statement holds with probability at least $n^{-6}$. (Of course, the constant $6$ can be picked arbitrarily larger.) Moreover, the running of this algorithm is $\Order(n m)$. This completes the proof of \cref{lem:random-restrictions}, but it remains to prove the intermediate claim.
\end{proof}

\begin{proof}[Proof of \cref{lem:random-restrictions:clm:single}]
Fix any pair $c_1, c_2 \in A + B$ and let~\makebox{$W = A \cap (\set{c_1, c_2} - B)$} denote the set of \emph{witnesses} of $c_1$ and $c_2$. As a first step, we prove that ~\makebox{$|W \cap A'| = 1$} holds with good probability. Indeed, if $|W| = 1$, then with probability at least $\frac{1}{L}$ we sample $\ell = 0$, and thus have $A' = A$ and~\makebox{$|W \cap A'| = 1$}. Suppose that $|W| \geq 2$ from now on. Then with probability $\frac{1}{L}$ we sample exactly~\makebox{$\ell = \ceil{\log |W|}$}. We condition on this event. Then, observing that $A' \cap W$ is a uniform subsample of $W$ with rate $2^{-\ell}$, we can bound
\begin{equation*}
	\Pr[|W \cap A'| = 1] = |W| \cdot (1 - 2^{-\ell})^{|W| - 1} \cdot 2^{-\ell} \geq |W| \cdot (1 - \tfrac{1}{|W|})^{|W|} \cdot \tfrac{1}{2|W|} \geq \tfrac{1}{8}.
\end{equation*}
Here we used that $f(x) = (1 - \frac{1}{x})^{x}$ is an increasing function with $f(2) = \frac14$. In summary, with probability at least $\frac{1}{8L}$, the set $A'$ satisfies that $|W \cap A'| = 1$. Condition on this event from now on, and let $a \in W \cap A'$ denote the unique witness.

We distinguish two cases: On the one hand, suppose that $a$ is not a witness of both~$c_1$ and~$c_2$. Without loss of generality assume that $a$ is a witness of $c_1$ (and not of $c_2$), and let $b_1 \in B$ be such that~\makebox{$a + b_1 = c_1$}. Then with probability~$\frac12$ we include $b_1$ into $B'$. If this happens, then indeed $(A' + B') \cap \set{c_1, c_2} = \set{c_1}$ and the claim follows. So assume, on the other hand, that $a$ is a witness of both $c_1$ and $c_2$. Let $b_1, b_2 \in B$ be such that $a + b_1 = c_1$ and $a + b_2 = c_2$. Then with probability exactly~$\frac12$ we include exactly one of the two elements $b_1, b_2$ into $B'$, i.e.~$|B' \cap \set{b_1, b_2}| = 1$. We conclude that $|(A' + B') \cap \set{c_1, c_2}| = 1$ as claimed.
\end{proof}

\thmsumset*
\begin{proof}
The algorithm is as follows: First, we precompute the size $t = |A + B|$ by \cref{lem:sumset-size}. Then we apply \cref{lem:random-restrictions} to compute subsets $A_1, \dots, A_m \subseteq A$ and $B_1, \dots, B_m \subseteq B$. For each $i \in [m]$, we apply \cref{lem:sumset-minimal-poly} (with the size bound $t$) to construct the polynomial
\begin{equation*}
    \Lambda_i(X) := \Lambda_{A_i + B_i}(X) = \prod_{c \in A_i + B_i} (X - c).
\end{equation*}
Finally, we apply \cref{thm:coprime-basis} to simultaneously factor the polynomials $\mathcal P = \set{\Lambda_1, \dots, \Lambda_m}$ into a collection of pairwise coprime polynomials $\mathcal Q$. We claim that we can simply read off the set $A + B$ from the factor polynomials $\mathcal Q$:

\begin{claim} \label{thm:sumset:clm:factors}
With high probability, $\mathcal Q = \set{X - c : c \in A + B}$.
\end{claim}
\begin{proof}[Proof of \cref{thm:sumset:clm:factors}]
Let $\mathcal Q = \set{Q_1, \dots, Q_{|\mathcal Q|}}$. First, note that since each polynomial $Q_j$ appears as a factor of some polynomial $\Lambda_i$, each such polynomial has the form
\begin{equation*}
    Q_j(X) = \prod_{c \in C_i} (X - c),
\end{equation*}
for some set $C_j \subseteq \Real$. In fact, since any polynomial $\Lambda_i$ has roots only in $A + B$, we have $C_j \subseteq A + B$. The proof would be complete if all sets~$C_j$ were singletons, so suppose that there is some set~$C_j$ containing two distinct elements $c_1, c_2 \in C_i \subseteq A + B$. By \cref{lem:random-restrictions}, with high probability there is some index $i$ such that $|(A_i + B_i) \cap \set{c_1, c_2}| = 1$. Assume without loss of generality that~\makebox{$c_1 \in A_i + B_i$} and $c_2 \not\in A_i + B_i$. Now consider the factorization of $\Lambda_i$ into a product of polynomials from $\mathcal Q$. Clearly,~$Q_j$ does not occur in the factorization (as $\Lambda_i$ is missing the factor $X - c_2$). However, this means that there is another polynomial in $\mathcal Q$ containing the factor $X - c_1$, contradicting the assumption that all polynomials in $\mathcal Q$ are pairwise coprime.
\end{proof}

Let us come back to the proof of \cref{thm:sumset} and complete the running time analysis. Precomputing $t$ by \cref{lem:sumset-size} takes time $\Order(t \log^2 t)$. The call to \cref{lem:random-restrictions} takes time $\Order(m t)$, where $m = \Order(\log^2 t)$, and the subsequent calls to \cref{lem:sumset-minimal-poly} take time $\Order(m t \log^2 t) = \Order(t \log^4 t)$. Finally, running \cref{thm:coprime-basis} takes time $\widetilde\Order(m t)$, so we can all in all bound the running time by $\widetilde\Order(t)$.
\end{proof}

\begin{theorem}[Real Convolution] \label{thm:real-conv}
There is a Las Vegas algorithm that, given $f, g : \Real \to \Real_{\geq 0}$ (with finite support), computes the convolution $f \conv g$ in expected time $\widetilde\Order(|\supp(f \conv g)|)$ (and with high probability).
\end{theorem}
\begin{proof}
We first compute the support $\supp(f \conv g) = \supp(f) + \supp(g)$ by \cref{thm:sumset} (here we use that $f$ takes only nonnegative values to avoid cancelations). Then, let $t = |\supp(f \conv g)|$ and assume that we have access to $\Sigma^0(f \conv g), \dots, \Sigma^{2t-1}(f \conv g)$ in the same way as before, using \cref{lem:eval-op,lem:conv-op}. By another application of \cref{lem:eval-op} we can interpolate the function $f \conv g$ from $\supp(f \conv g)$ and $\Sigma^0(f \conv g), \dots, \Sigma^{2t-1}(f \conv g)$. The running time is bounded by $\widetilde\Order(t)$.
\end{proof}
\section{Real Constellation in Near-Linear Time} \label{sec:constellation}
In this section we design our near-linear time algorithm for the Constellation problem. The algorithm is inspired by Fischer's recent \emph{deterministic} Constellation algorithm, and relies on some more algebraic tools, in particular the Baur-Strassen theorem.

\subsection{The Baur-Strassen Theorem}
The Baur-Strassen theorem intuitively states that any algebraic algorithm computing some function~$f$ can compute additionally all partial derivatives of $f$, essentially for free. To make this precise, we first have to fix what we understand by an algebraic algorithm.

\begin{definition}[Arithmetic Circuit]
An \emph{arithmetic circuit $C$} over the field $\Field$ and the variables $x_1, \dots, x_n$ is a directed acyclic graph as follows. The nodes are called \emph{gates}, and are of the following two types: Each gate either has in-degree $0$ and is labeled with a variable $x_i$ or a constant $\alpha \in \Field$, and it has in-degree $2$ and is labeled with an arithmetic operation (${+}, {-}, {\times}, {/}$).
\end{definition}

We refer to the gates labeled by variables $X_i$ as \emph{input gates}, by constants as \emph{constant gates} and by an operation ${\circ} \in \set{{+}, {-}, {\times}, {/}}$ as \emph{$\circ$-gates}. A gate with out-degree $0$ is called an \emph{output gate.} The \emph{size} of an arithmetic circuit $C$, denoted by $|C|$, is the number of gates plus number of edges in $C$. Note that each gate in an arithmetic circuit computes a polynomial $P \in \Field[x_1, \dots, x_n]$ in a natural way: Input gates compute~$x_i$, constant gates compute the constant polynomial $\alpha$, and each, say, $\times$-gate computes the product of the polynomials computed by its two incoming gates. We typically say that an arithmetic circuit computes polynomials~\makebox{$P_1, \dots, P_m$} if there are $m$ output gates computing these respective polynomials.

\begin{lemma}[Baur-Strassen Theorem,~\cite{BaurS83,Morgenstern85}]
For any arithmetic circuit $C$ computing $f(x_1, \dots, x_n)$, there is a circuit $C'$ that simultaneously computes the partial derivatives $\frac{\partial f}{\partial x_i}(x_1, \dots, x_n)$ for all $1 \leq i \neq n$. The circuit $C'$ has size $|C'| \leq \Order(|C|)$ and can be constructed in time $\Order(|C|)$.
\end{lemma}

Arithmetic circuits form a natural model of algebraic computation. It can easily be verified that most algorithms developed in the last section can similarly be implemented with arithmetic circuits. In particular, it is relevant that \cref{lem:eval-op,lem:conv-op} construct arithmetic circuits in time $\Order(t \log^2 t)$ that implement exactly the evaluation, interpolation or convolution operations.\footnote{More precisely, \cref{lem:eval-op} (say) can be seen as follows: There is an algorithm running in time $\Order(t \log^2 t)$ that produces an arithmetic circuit of size $\Order(t \log^2 t)$ with inputs $f(a_0), \dots, f(a_{t-1})$ (where $a_0, \dots, a_{t-1}$ is the support of~$f$) and outputs $\Sigma^0(f), \dots, \Sigma^{t-1}(f)$.}

\subsection{The Algorithm}
The following lemma is similar to~\cite[Theorem~4.1]{Fischer24}:

\begin{lemma} \label{lem:constellation-cand}
Let $A, B, S \subseteq \Real$. There is a Las Vegas algorithm computing, for each $s \in S$, whether $A + s \subseteq B$ in time $\widetilde\Order(|A + S| + |B|)$.
\end{lemma}
\begin{proof}
We start with a description of the algorithm. First, we precompute the sumset $A + S$ by \cref{thm:sumset}, and let $t = |A + S|$. Then, we construct an arithmetic circuit with inputs $x_a$ and $y_s$ (for $a \in A$ and $s \in S$) and outputs $z_b$ (for $b \in A + S$) defined by
\begin{equation*}
    z_b = \sum_{\substack{a \in A\\s \in S\\a + s = b}} x_a \cdot y_s.
\end{equation*}
This circuit can be constructed by concatenating the evaluation, convolution and interpolation circuits from \cref{lem:eval-op,lem:conv-op}. Specifically, the circuit first computes $\Sigma^0(\One_A), \dots, \Sigma^{t-1}(\One_A)$ and $\Sigma^0(\One_S), \dots, \Sigma^{t-1}(\One_S)$ by \cref{lem:eval-op}, then computes $\Sigma^0(\One_A \conv \One_S), \dots, \Sigma^{t-1}(\One_A \conv \One_S)$ by \cref{lem:conv-op}, and then interpolates $z_b = (\One_A \conv \One_S)(b)$ by \cref{lem:eval-op}. We will now augment this arithmetic circuit in some more steps.

First, add another output gate
\begin{equation*}
    z := \sum_{b \in (A + S) \cap B} z_b
\end{equation*}
to the circuit. Then apply the Baur-Strassen theorem to let the circuit compute in addition all partial derivatives $\frac{\partial z}{\partial y_s}$. Note that:
\begin{equation*}
    \frac{\partial z}{\partial y_s} = \sum_{b \in (A + S) \cap B} \frac{\partial z_b}{\partial y_s} = \sum_{b \in (A + S) \cap B} \sum_{\substack{a \in A\\a + s=b}} x_a = \sum_{\substack{a \in A\\b \in B\\a + s = b}} x_a.
\end{equation*}
Therefore, by plugging in $x_a \gets 1$, the output $\frac{\partial z}{\partial y_s}$ equals $|A|$ if and only if $A + s \subseteq B$.

It remains to bound the running time. The applications of \cref{lem:eval-op,lem:conv-op} takes time $\Order(t \log^2 t)$. Moreover, the construction of the arithmetic circuit, including all modifications run in linear time in the circuit size $\Order(t \log^2 t)$. Similarly, evaluating the circuit on the all-ones input runs in linear time in the circuit size.
\end{proof}

\thmconstellation*
\begin{proof}
Let $n = |A| + |B|$ and let $L = \ceil{\log n}$. The algorithm proceeds in levels $\ell \gets L, \dots, 1, 0$. Throughout, we maintain sets of candidate shifts $S_{L+1}, S_L, \dots, S_0$ such that ultimately $S_0 = S$. Initially, let $a \in A$ be arbitrary and set $S_{L+1} \gets B - a$; clearly~$S_{L+1}$ contains all feasible shifts $s$ with $A + s \subseteq B$ (as in particular $a + s \in B$). Each level $\ell$ runs the following steps:
\begin{enumerate}
    \item Uniformly subsample a set $A_\ell \subseteq A$ with rate $2^{-\ell}$.
    \item Let $S_\ell \gets \set{s \in S_{\ell+1} : A_\ell + s \subseteq B}$ as computed by \cref{lem:constellation-cand}.
\end{enumerate}
Finally, report $S = S_0$. This completes the description of the algorithm.

We turn to the correctness analysis and prove that $S_0 = S = \set{s \in \Real : A + s \subseteq B}$. On the one hand, it is easy to verify by induction that $S \subseteq S_\ell$ for all levels $\ell$, and thus in particular $S \subseteq S_0$. On the other hand, since $A_0 = A$, we have $S_0 \supseteq \set{s \in S : A + s \subseteq B} = S$. 

Let us finally consider the running time. The algorithm runs in $\Order(\log n)$ levels, and preparing the sets $A_\ell$ takes time $\widetilde\Order(n)$. The dominant contribution to the running time is $\widetilde\Order(\sum_\ell |A_\ell + S_{\ell+1}| + |B|)$ due to \cref{lem:constellation-cand}. Our goal in the following is therefore to bound $|A_\ell + S_{\ell+1}|$ for any level $\ell$. To this end, fix any shift $s \in S_{L+1}$ and note that
\begin{equation*}
    \Pr_{A_L, \dots, A_\ell}[s \in S_\ell] \leq (1 - 2^{-\ell})^{|(A + s) \setminus B|} \leq \exp(-2^{-\ell} \cdot |(A + s) \setminus B|)
\end{equation*}
(where the randomness is over all levels down to $\ell$). Therefore:
\begin{align*}
    \Ex_{A_L, \dots, A_\ell} |A_\ell + S_{\ell+1}|
    &\leq |B| + \Ex_{A_L, \dots, A_\ell} |(A_\ell + S_{\ell+1}) \setminus B| \\
    &\leq |B| + \sum_{s \in S_{L+1}} \Pr_{A_L, \dots, A_{\ell+1}}[s \in S_{\ell+1}] \cdot \Ex_{A_\ell} |(A_\ell + s) \setminus B| \\
    &\leq |B| + \sum_{s \in S_{L+1}} \exp(-2^{-\ell-1} \cdot |(A + s) \setminus B|) \cdot 2^{-\ell} \cdot |(A + s) \setminus B| \\
    &\leq |B| + |S_{L+1}| \cdot \max_{x \geq 0} (\exp(-x) \cdot 2x) \\
    &\leq (1 + 2/e) n.
\end{align*}
Here, in the second-to-last step we have substituted $x = 2^{-\ell} \cdot |(A + s) \setminus B|$ and then used the fact that $\exp(-x) \cdot x \leq 1/e$ for all $x \geq 0$. Summing over the $\Order(\log n)$ levels, the total expected running time is indeed $\widetilde\Order(n)$. (Using a standard interrupt-and-repeat argument we can make the algorithm run in time $\widetilde\Order(n)$ with high probability at the cost of one additional $\log n$-factor.)
\end{proof}

\appendix
\section{Factoring Polynomials into Coprimes} \label{sec:coprime}
In this section we revisit the following theorem due to Bernstein~\cite{Bernstein05}. Bernstein analyzes \cref{thm:coprime-basis} over general domains supporting the gcd operation, and assuming a computational model other than the word or real RAM model. We here include a complete proof, in which we limit our attention to polynomials and the word RAM model equipped with an oracle performing field operations in unit time (e.g.~the real RAM model).

\thmcoprimebasis*

Throughout this section we implicitly use that basic operations on polynomials such as computing products, divisions with remainder and greatest common divisors have classic near-linear time algorithms over any field; see e.g.~\cite{vonzurGathenG13}.

Note that we can assume that all polynomials $P$ and $Q$ in \cref{thm:coprime-basis} are monic (i.e., have leading coefficient~$1$). For the scope of this section, let us introduce some more notation: We say that a set of polynomials~$\mathcal Q$ is a \emph{coprime basis} of some other set of polynomials $\mathcal P$ if the polynomials~$\mathcal Q$ are pairwise coprime, and every polynomial $P \in \mathcal P$ can be factored into polynomials from~$\mathcal Q$. We say that the coprime basis is \emph{reduced} if every polynomial $Q \in \mathcal Q$ appears as a factor of some polynomial~\makebox{$P \in \mathcal P$}. For a set of polynomials $\mathcal P$ we additionally write $\deg(\mathcal P) = \sum_{P \in \mathcal P} \deg(P)$. The following two observations will be useful:

\begin{observation}
If $\mathcal Q$ is a coprime basis of $\mathcal P$, then for any subset $\mathcal P' \subseteq \mathcal P$, $\gcd(\mathcal P')$ can be factored into polynomials from $\mathcal Q$. 
\end{observation}
\begin{proof}
It suffices to check that the set of polynomials that can be factored into polynomials from~$\mathcal Q$ is closed under taking greatest common divisors. To this end, consider any two polynomials expressible as \smash{$\prod_{Q \in \mathcal Q} Q^{a_Q}$} and \smash{$\prod_{Q \in \mathcal Q} Q^{b_Q}$}, respectively. Their gcd is \smash{$\prod_{Q \in \mathcal Q} Q^{\min(a_Q, b_Q)}$}, which clearly also factors into polynomials from $\mathcal Q$.
\end{proof}

\begin{observation} \label{obs:coprime-basis-reduced}
If $\mathcal Q$ is a reduced coprime basis of $\mathcal P$, then $\deg(\mathcal Q) \leq \deg(\mathcal P)$.
\end{observation}
\begin{proof}
Consider the factorization of $\prod(\mathcal P) := \prod_{P \in \mathcal P} P$ into polynomials from $\mathcal Q$. Since $\mathcal Q$ is pairwise coprime, this factorization is unique. Since $\mathcal Q$ is further reduced, every polynomial $Q \in \mathcal Q$ appears at least once. Thus, $\deg(\mathcal Q) = \deg(\prod(\mathcal Q)) \leq \deg(\prod(\mathcal P)) = \deg(\mathcal P)$.
\end{proof}

In light of the previous \cref{obs:coprime-basis-reduced}, to prove \cref{thm:coprime-basis} our goal is to give an algorithm computing a reduced coprime basis $\mathcal Q$ of~\makebox{$\mathcal P$}. Following the approach in~\cite{Bernstein05}, we establish the algorithm in three steps: First, we show how to efficiently extend a coprime basis by one additional polynomial (\cref{lem:coprime-basis-extend}). Second, we show how to efficiently merge two coprime bases (\cref{lem:coprime-basis-merge}). Third, we prove \cref{thm:coprime-basis} via a simple divide-and-conquer strategy. 

\begin{lemma}[Extending a Coprime Basis] \label{lem:coprime-basis-extend}
Let $P$ be a polynomial, let $\mathcal Q$ be a reduced coprime basis of some polynomials~$\mathcal P$. Given $P$ and $\mathcal Q$, in time~\smash{$\widetilde\Order(\deg(P) + \deg(\mathcal Q))$} we can compute a reduced coprime basis $\mathcal Q'$ of~\makebox{$\mathcal P \cup \set P$}.
\end{lemma}
\begin{proof}
Our goal is to compute
\begin{equation*}
    \mathcal Q' = \set*{\gcd(P, Q), \frac{Q}{\gcd(P, Q)} : Q \in \mathcal Q} \cup \set*{\frac{P}{\gcd(P, \prod(\mathcal Q))}}
\end{equation*}
It is easy to check that $\mathcal Q'$ is a coprime basis of $\mathcal P \cup \set{P}$: All the polynomials in $\mathcal Q$ as well as $P$ can be factored into elements from $\mathcal Q'$, and any pair of polynomials in $\mathcal Q'$ is coprime. Moreover, $\mathcal Q'$ is reduced (assuming that $\mathcal Q$ is reduced).

Towards computing $\mathcal Q'$, we first compute $\prod(\mathcal Q)$. For the sake of efficiency, we use the following divide-and-conquer algorithm. Let $\mathcal Q = \set{Q_0, \dots, Q_{n-1}}$ and assume for simplicity that $n$ is a power of~$2$. We can compute for each dyadic interval\footnote{Recall that a dyadic interval $I \subseteq [n]$ is an interval of the form $[k \cdot 2^i, (k+1) \cdot 2^i)$ for some $i \in [\log n]$ and $k \in [n / 2^i]$.} $I \subseteq [n]$ the product $\prod_{i \in I} Q_i$ in a binary tree from bottom to top: At the leaves $\set{1}, \dots, \set{n}$ the desired products are equal to the polynomials~\makebox{$Q_1, \dots, Q_n$}, at each internal node we compute the product of the polynomials of the respective two children, and at the root we obtain $\prod(\mathcal Q) = \prod_{i \in [n]} Q_i$. The total running time of any layer in the tree is~\smash{$\widetilde\Order(\sum_{i \in [n]} \deg(Q_i)) = \widetilde\Order(\deg(\mathcal Q))$}, and the tree has depth $\log n \leq \log \deg(\mathcal Q)$.

Next, we compute the polynomials $P \bmod Q_I$ for each dyadic interval $I \subseteq [n]$. To this end, we traverse the same binary tree as in the previous step---but this time from top to bottom. At the root we simply compute $P \bmod Q_{[n]}$. Focus next on any node corresponding to the dyadic interval~\makebox{$I = I_1 \cup I_2$} for which we have already computed $P \bmod Q_I$. To deal with its children labeled with $I_1$ and $I_2$, we compute~\makebox{$P \bmod Q_{I_1} = (P \bmod Q_I) \bmod Q_{I_1}$} and similarly for $P \bmod Q_{I_2}$. At each level of the tree we perform operations on polynomials of total degree $\Order(\sum_{i \in [n]} \deg(Q_i)) = \Order(\deg(\mathcal Q))$, thus the total time of this step is also~\smash{$\widetilde\Order(\deg(\mathcal Q))$}.

Finally, we compute the polynomials (i)~$\gcd(P, Q_i) = \gcd(P \bmod Q_i, Q_i)$ and (ii)~$Q_i / \gcd(P, Q_i)$ for all $i \in [n]$, and~(iii) $P / \gcd(P, \prod(\mathcal Q))$. This requires (i)~$n$ gcd computations, plus (ii)~$n$ divisions with remainder, plus~(iii) one division with remainder, each of which takes time $\widetilde\Order(\deg(\mathcal Q))$.
\end{proof}

\begin{lemma}[Merging Two Comprime Bases] \label{lem:coprime-basis-merge}
Let $\mathcal Q_1$ be a reduced coprime basis of $\mathcal P_1$, and let~$\mathcal Q_2$ be a reduced coprime basis of $\mathcal P_2$. Given $\mathcal Q_1$ and $\mathcal Q_2$, in time~\smash{$\widetilde\Order(\deg(\mathcal Q_1) + \deg(\mathcal Q_2))$} we can compute a reduced coprime basis $\mathcal Q'$ of $\mathcal P_1 \cup \mathcal P_2$.
\end{lemma}
\begin{proof}
Let $L = \ceil{\log |\mathcal Q_2|}$. Consider an arbitrary ``naming'' of the polynomials in $\mathcal Q_2$ by length-$L$ bit-strings. Formally, let $\phi : \mathcal Q_2 \to \set{0, 1}^L$ be an arbitrary injection. Then consider the following polynomials for $\ell \in [L]$ and $b \in \set{0, 1}$:
\begin{equation*}
    R_{\ell, b} = \prod_{\substack{Q \in \mathcal Q\\\phi(Q)[\ell] = b}} Q.
\end{equation*}
We can compute each such polynomial in time $\widetilde\Order(\deg(\mathcal Q_2))$ in the same divide-and-conquer fashion as before. We then compute a reduced coprime basis $\mathcal Q'$ of $\mathcal P_1 \cup \set{R_{\ell, b} : \ell \in [L],\, b \in \set{0, 1}}$ by applying the previous \cref{lem:coprime-basis-extend} repeatedly $2L$ times. The running time of each call is bounded by~\smash{$\widetilde\Order(\deg(\mathcal Q_1) + L \deg(\mathcal Q_2))$} and thus~\smash{$\widetilde\Order(L \deg(\mathcal Q_1) + L^2 \deg(\mathcal Q_2)) = \widetilde\Order(\deg(\mathcal Q_1) + \deg(\mathcal Q_2))$} in total.

It remains to prove that $\mathcal Q'$ is indeed a reduced coprime basis for $\mathcal P_1 \cup \mathcal P_2$. Clearly $\mathcal Q'$ is pairwise coprime, and we can factor every polynomial in $\mathcal P_1$ into polynomials from $\mathcal Q'$. It remains to check that every polynomial $\mathcal P_2$ can similarly be factored. Since $\mathcal Q_2$ is a coprime basis, it suffices to prove that every $Q \in \mathcal Q_2$ can be factored into polynomials from $\mathcal Q'$. To this end, note that
\begin{equation*}
    \gcd(\,\set{R_{\ell, \phi(Q)[\ell]} : \ell \in [L]}\,) = \gcd\parens*{\,\set*{\prod_{\substack{Q' \in \mathcal Q\\\phi(Q')[\ell] = \phi(Q)[\ell]}} : \ell \in [L]}\,} = Q,
\end{equation*}
using that any other polynomial $Q' \in \mathcal Q_2$ appears in at least one product with $\phi(Q')[i] \neq \phi(Q)[i]$.
\end{proof}

\begin{proof}[Proof of \cref{thm:coprime-basis}]
We follow a simple divide-and-conquer algorithm: If $\mathcal P$ contains a single polynomial $P$, then we simply return $\mathcal Q = \set{P}$. Otherwise, we partition $\mathcal P$ arbitrarily into two halves~\makebox{$\mathcal P = \mathcal P_1 \cup \mathcal P_2$} (of roughly the same size), and recursively compute reduced coprime bases~$\mathcal Q_1$ and $\mathcal Q_2$ of $\mathcal P_1$ and $\mathcal P_2$, respectively. We merge $Q_1$ and $Q_2$ by one call to \cref{lem:coprime-basis-merge}. \cref{lem:coprime-basis-merge} runs in time $\widetilde\Order(\deg(\mathcal P))$, and the recursion worsens the running time only by $\Order(\log\deg(\mathcal P))$.
\end{proof}
\section{Real 3SUM with Preprocessing} \label{sec:3sum}
In this short section we provide the proof of \cref{thm:3sum-preprocessing}. 

\thmthreesumpreprocessing*

We reuse Chan and Lewenstein's~\cite{ChanL15} proof without any modifications other than plugging in our new algorithm for computing real sumsets. It is based on the corollary of the BSG theorem: 

\begin{theorem}[BSG Corollary, \cite{ChanL15}] \label{thm:bsg-corollary}
Let $A, B, C$ be subsets of an Abelian group of size at most~$n$. For any $0 < \alpha < 1$, there exist subsets $A_1, \dots, A_k \subseteq A$ and $B_1, \dots, B_k \subseteq B$ such that
\begin{itemize}
    \item the remainder set $R = \set{(a, b) \in A \times B : a + b \in C} \setminus \bigcup_{i=1}^k (A_i \times B_i)$ has size at most $\alpha n^2$.
    \item $|A_1 + B_1|, \dots, |A_k + B_k| \leq \Order(n / \alpha^5)$.
    \item $k = \Order(1/\alpha)$.
\end{itemize}
Moreover, we can compute $A_1, \dots, A_k, B_1, \dots, B_k$ and $R$ in Las Vegas time $\widetilde\Order(n^2)$.
\end{theorem}

\begin{proof}[Proof of \cref{thm:3sum-preprocessing}]
Let $0 < \alpha < 1$ be a parameter to be determined later. In the preprocessing phase we apply the BSG Corollary to $(A, B, C)$ to compute in time $\widetilde\Order(n^2)$ the sets~$A_1, \dots, A_k$, $B_1, \dots, B_k$ and $R$. We store these sets for later use. This takes space $\Order(k n + \alpha n^2)$.

In the query phase, given $A' \subseteq A, B' \subseteq B, C' \subseteq C$, we first test whether there is a pair $(a, b) \in R$ with $a \in A', b \in B'$ and $a + b \in C'$, and report ``yes'' in this case. Otherwise, we use \cref{thm:sumset} to compute the sumsets $(A_i \cap A') + (B_i \cap B')$ for $i \gets 1, \dots, k$, and test whether they respectively intersect $C_i \cap A'$. If no 3SUM solution was found in this way, we stop and report~``no''. The correctness is obvious from the first property of the BSG Corollary. Furthermore, the running time of any query is bounded by
\begin{equation*}
    \widetilde\Order\parens*{|R| + \sum_{i=1}^k |A_i + B_i|} = \widetilde\Order(\alpha n^2 + k \cdot n / \alpha^5) = \widetilde\Order(\alpha n^2 + n / \alpha^6).
\end{equation*}
The proof is complete by choosing $\alpha = n^{-1/7}$.
\end{proof}
\section{Real Restricted Sumsets} \label{sec:sumset-restricted}
The purpose of this section is to formally prove that we can compute interval- and prefix-restricted sumsets of real numbers in the same running time obtained by Bringmann and Nakos~\cite{BringmannN20} for integer sets.

\paragraph{Interval-Restricted Sumsets}
We start with the result for interval-restricted sumsets:

\thmintervalsumset*
\begin{proof}
Let $g$ be a parameter to be fixed later. We partition $A = A_1 \sqcup \dots \sqcup A_g$ and $B = B_1 \sqcup \dots \sqcup B_g$ into subsets of size at most $\ceil{|A| / g}$ and $\ceil{|B| / g}$, respectively, and in such a way that
\begin{align*}
    \max(A_1) &< \min(A_2),\quad \dots,\quad \max(A_{g-1}) < \min(A_g), \\
    \max(B_1) &< \min(B_2),\quad \dots,\quad \max(B_{g-1}) < \min(B_g).
\end{align*}
We call a pair $(i, j) \in [g]^2$ \emph{partially relevant} if~\makebox{$[\min(A_i) + \min(B_j), \max(A_i) + \max(B_j)] \cap [\ell, u] \neq \emptyset$}, and \emph{totally relevant} if further $[\min(A_i) + \min(B_j), \max(A_i) + \max(B_j)] \subseteq [\ell, u]$. For each partially relevant pair $i, j \in [g]$ we compute the sumset $A_i + B_j$ in output-sensitive time using \cref{thm:sumset}, and include $(A_i + B_j) \cap [\ell, u]$ in the output.

The correctness is clear. To analyze the running time, we call $\mathcal{C}_\Delta = \set{(i, j) \in [g]^2 : i - j = \Delta}$ a chain; note that the only nonempty chains are $C_{-g+1}, \dots, C_{g-1}$. Each chain may contain several partially relevant pairs, but at most two of these are not totally relevant. Thus, we can bound
\begin{align*}
    \sum_{\substack{(i, j) \in [g]^2\\\text{part.\ rel.}}} |A_i + B_j|
    &= \sum_{\Delta=-g+1}^{g+1} \sum_{\substack{(i, j) \in \mathcal{C}_\Delta\\\text{part.\ rel.}}} |A_i + B_j| \\
    &\leq \sum_{\Delta=-g+1}^{g+1} \parens*{2 \ceil*{\frac{|A|}{g}} \ceil*{\frac{|B|}{g}} + \sum_{\substack{(i, j) \in \mathcal{C}_\Delta\\\text{tot.\ rel.}}} |A_i + B_j|} \\
    &\leq \sum_{\Delta=-g+1}^{g+1} \parens*{2 \ceil*{\frac{|A|}{g}} \ceil*{\frac{|B|}{g}} + |C|} \\
    &= \Order\parens*{\frac{|A| \, |B|}{g} + g |C|}.
\end{align*}
This is optimized by setting $g = \sqrt{|A| \, |B| / |C|}$. Of course, we have no knowledge of $|C|$ in advance, so instead we run $\log(|A| \, |B|)$ independent copies of our algorithm with exponentially increasing guesses of~$|C|$.
\end{proof}

\paragraph{Prefix-Restricted Sumsets}
Next, we consider prefix-restricted sumsets (i.e., where the left bound of the interval is $-\infty$).

\thmprefixsumset*

The proof relies on the following structural lemma which is an immediate consequence of applying Ruzsa's triangle inequality~\cite{Ruzsa96} twice:

\begin{lemma}[Corollary of Ruzsa's Triangle Inequality, \cite{Ruzsa96,BringmannN20}] \label{lem:ruzsa-corollary}
Let $A, B, C, D$ be subsets of an Abelian group. Then:
\begin{equation*}
    |A + B| \leq \frac{|A + C| \cdot |C + D| \cdot |D + B|}{|C| \cdot |D|}.
\end{equation*}
\end{lemma}

\begin{lemma}[Sumset with Time Budget] \label{lem:sumset-budget}
There is a randomized algorithm that, given $A, B \subseteq \Real$ and an integer $s \geq 1$, either (correctly) computes the set $A + B$, or claims (with high probability) that $|A + B| \geq s$. The algorithm runs in worst-case time $\widetilde\Order(s)$.
\end{lemma}
\begin{proof}
We simply run \cref{thm:sumset} to compute the sumset $A + B$, but interrupt the computation if it exceeds $\widetilde\Theta(s)$ computation steps (for some appropriate hidden log-factor overhead depending on \cref{thm:sumset}); in this case we report that $|A + B| \geq s$. Indeed, if $|A + B| \leq s$, then \cref{thm:sumset} computes the sumset in time $\widetilde\Order(s)$ with high probability. Moreover, if \cref{thm:sumset} terminates, then the sumset is certainly correct.
\end{proof}

\begin{lemma} \label{lem:quad-tree-heavy}
Let $A, B \subseteq [0, u]$, let $C = (A + B) \cap [0, u]$, and consider partitions $A_1 \sqcup \dots \sqcup A_g$ and $B = B_1 \sqcup \dots \sqcup B_g$ satisfying the following four conditions:
\begin{enumerate}[label=(\roman*)]
    \item $\max(A_1) < \min(A_2), \dots, \max(A_{g-1}) < \min(A_g)$, and\\$\min(B_1) > \max(B_2), \dots, \min(B_{g-1}) > \max(B_g)$.
    \item $A_i + B_j \subseteq [0, u]$ for all $i < j$.
    \item $|A_1|, \dots, |A_g| \in [n, 2n]$ and $|B_1|, \dots, |B_g| \in [m, 2m]$ for some integers $n, m$.
    \item $|A_\ell + B_\ell| > |C|$.
\end{enumerate}
Then $g < 258 |C|^{1/3}$.
\end{lemma}

Let us assume for the moment that \cref{lem:quad-tree-heavy} is given, and first focus on completing the proof of \cref{thm:prefix-sumset}. We provide the proof of \cref{lem:quad-tree-heavy} right after.

\begin{proof}[Proof of \cref{thm:prefix-sumset}]
First, note that if $A$ contains negative elements, then we can simply shift~$A$ and $u$ by $\min(A)$ (and later apply the reverse transformation to $C$). We can similarly assume that~$B$ contains only nonnegative entries. After removing all elements bigger than $u$ from $A$ and~$B$, this preprocessing ensures that $A, B \subseteq [0, u]$.

Suppose for now that we are given an estimate $|C| \leq s \leq 2|C|$; we will later remove this assumption. Consider the following recursive algorithm: We first attempt to compute~\makebox{$A + B$} by \cref{lem:sumset-budget} with parameter $s$. If the lemma succeeds we return~\makebox{$(A + B) \cap [0, u]$}. Otherwise, if the lemma claims that $|A + B| > s$, partition $A = A_1 \sqcup A_2$ such that~\smash{$|A_1|, |A_2| \leq \ceil{\frac{|A|}{2}}$} and~\makebox{$\max(A_1) < \min(A_2)$}. Next, partition $B = B_1 \sqcup B_2$ such that~\makebox{$B_1 = B \cap [0, u - \max(A_1)]$}. We then compute $A_1 + B_1$ in near-linear time by \cref{thm:sumset}, and compute $(A_1 + B_2) \cap [0, u]$ and~\makebox{$(A_2 + B_1) \cap [0, u]$} recursively. The output is the union of these three sets. (We can discard the missing set $A_2 + B_2$ as it exceeds the threshold $u$.)

The correctness is immediate, so we focus on the running time. Ignoring the cost of recursive calls, the algorithm takes time $\widetilde\Order(s)$ (by \cref{lem:sumset-budget}) plus $\widetilde\Order(|A_1 + B_1|)$ (by \cref{thm:sumset}). As~\makebox{$s \leq 2|C|$} and since $A_1 + B_1 \subseteq [0, u]$ (and thus $A_1 + B_1 \subseteq C$), both contributions are bounded by~\smash{$\widetilde\Order(|C|)$}. To obtain the claimed bound on the running time, we therefore prove in the following that the number of recursive calls is bounded by $\Order(|C|^{1/3} \log |B|)$.

Suppose otherwise that at some level the algorithm reaches more than $516 |C|^{1/3} \cdot \log |B|$ recursive calls. Then in the previous level there are at least $258 |C|^{1/3} \cdot \log |B|$ recursive calls that are not immediately solved by the call to \cref{lem:sumset-budget}. Let $(A_1, B_1), \dots, (A_{g'}, B_{g'})$ denote the inputs to these recursive calls. It is easy to show by induction that these inputs form a ``staircase'' in the sense that
\begin{align*}
    &\max(A_1) < \min(A_2),\quad \dots,\quad \max(A_{g'-1}) < \min(A_{g'}), \\
    &\min(B_1) > \max(B_2),\quad \dots,\quad \min(B_{g'-1}) > \max(B_{g'}),
\end{align*}
and that $A_i + B_j \subseteq [0, u]$ for all $i < j$. Moreover, with high probability none of the calls to \cref{lem:sumset-budget} have failed and we indeed have $|A_1 + B_1|, \dots, |A_{g'} + B_{g'}| > s \geq |C|$. All in all, this shows that the sets $A_1, \dots, A_{g'}, B_1, \dots, B_{g'}$ satisfy the conditions~(i),~(ii) and~(iv) of \cref{lem:quad-tree-heavy}. To also satisfy condition~(iii), note on the one hand that all sets $A_1, \dots, A_{g'}$ have the same size (up to $\pm 1$). On the other hand, by a pigeonhole argument there is some $m \in [|B|]$ such that at least~$g := 258 |C|^{1/3}$ sets from $B_1, \dots, B_{g'}$ have size at least $m$ and at most $2m$. We restrict our attention to these $g$ sets which satisfy all conditions from \cref{lem:quad-tree-heavy}, and thereby cause a contradiction.

We finally remove the assumption that an estimate $|C| \leq s \leq 2|C|$ is known in advance. To this end we simply run $L = \log(|A| + |B|)$ copies of the previous algorithm in parallel with estimates~\makebox{$s = 2^0, 2^1, \dots, 2^{L-1}$}. We stop as soon as one of these executions terminates (noting that the previous algorithm is Las Vegas and will only return the correct set $C$). The running time is $\widetilde\Order(|A| + |B| + |C|^{4/3})$.
\end{proof}

\begin{proof}[Proof of \cref{lem:quad-tree-heavy}]
For each $\ell \in [g]$, define
\begin{equation*}
    S(\ell) := \sum_{\substack{i, j \in [g]\\i < \ell < j}} |A_\ell + B_j| + |A_i + B_j| + |A_i + B_\ell|.
\end{equation*}
By comparing how often each term $|A_i + B_j|$ appears in the sum, we can bound
\begin{equation*}
    \sum_{\ell \in [g]} S(\ell) \leq 3g \cdot \sum_{\substack{i, j \in [g]\\i < j}} |A_i + B_j|.
\end{equation*}
To bound this expression further, let us write $\mathcal{C}_\Delta = \set{(i, j) \in [g]^2 : i + j = \Delta, i < j}$. From condition~(i) we infer that, for any fixed $\Delta$, the sets~\smash{$\set{A_i + B_j}_{(i, j) \in \mathcal{C}_\Delta}$} are pairwise disjoint. By condition~(ii), it follows that $\sum_{(i, j) \in \mathcal{C}_\Delta} |A_i + B_j| \leq |\bigcup_{(i, j) \in \mathcal{C}_\Delta} (A_i + B_j)| \leq |C|$, and therefore
\begin{equation*}
    \sum_{\ell \in [g]} \leq 3g \cdot \sum_{\substack{i, j \in [g]\\i < j}} |A_i + B_j| \leq 3g \cdot \sum_{\Delta=1}^{2g-1} \sum_{(i, j) \in \mathcal{C}_\Delta} |A_i + B_j| \leq 6g \cdot |C|.
\end{equation*}
In particular, there is some $\frac13 g < \ell < \frac23 g$ satisfying that $S(\ell) \leq 18g \cdot |C|$. Fixing this $\ell$, there are at least $\frac13 g$ values $i \in [g]$ with $i < \ell$, and similarly there are at least $\frac13 g$ values $j \in [g]$ with $\ell < j$. Thus, there exist $i < \ell < j$ such that
\begin{equation*}
    |A_\ell + B_j| + |A_i + B_j| + |A_i + B_\ell| \leq \frac{162g \cdot |C|}{g^2} = \frac{162 |C|}{g}.
\end{equation*}
By the corollary of Ruzsa's triangle inequality (\cref{lem:ruzsa-corollary}) we conclude that
\begin{equation*}
    |A_\ell + B_\ell| \leq \frac{|A_\ell + B_j| \cdot |A_i + B_j| \cdot |A_i + B_\ell|}{|A_i| \cdot |B_j|} \leq \frac{162^3 |C|^3}{g^3 n m},
\end{equation*}
using assumption~(iii) in the last step. Note that $4 n m \geq |A_\ell| \cdot |B_\ell| \geq |A_\ell + B_\ell|$. Therefore, and using assumption~(iv), we finally have that
\begin{equation*}
    g^3 \leq \frac{162^3 |C|^3}{|A_\ell + B_\ell| \cdot nm} \leq \frac{4 \cdot 162^3 |C|^3}{|A_\ell + B_\ell|^2} \leq 4 \cdot 162^3 |C|,
\end{equation*}
and $g \leq 4^{1/3} \cdot 162 |C|^{1/3} \leq 258 |C|^{1/3}$ follows.
\end{proof}
\section{Real Subset Sum} \label{sec:subsetsum}
In this section we prove that some algorithms for the Subset Sum problem can be adapted to the reals.

\paragraph{Subset Sum without a Target}
We start with the simpler problem of computing \emph{all} subset sums $\SSS(A) = \set{\Sigma(A') : A' \subseteq A}$ (without any target restriction). Over the integers it is folklore that $\SSS(A)$ can be computed in near-linear output-sensitive time by repeated sumset computations. We show that $\SSS(A)$ can similarly computed for the reals:

\begin{restatable}[Real Subset Sum without Target]{theorem}{thmsubsetsum} \label{thm:subset-sum}
There is a Las Vegas algorithm that, given a multiset $X \subseteq \Realnneg$, computes the set of subset sums $\mathcal S(X)$ in time~\smash{$\widetilde\Order(|\mathcal S(X)|)$}.
\end{restatable}

We start with two simple lower bounds on the size of (iterated) sumsets; see e.g.~\cite[Lemma~5.3]{TaoV06}:

\begin{lemma} \label{lem:sumset-lower-bound}
Let $A, B \subseteq \Real$. Then $|A + B| \geq |A| + |B| - 1$.
\end{lemma}
\begin{proof}
Let $a_1 < \dots < a_{|A|}$ and $b_1 < \dots < b_{|B|}$ denote the elements of $A$ and $B$ respectively. Then the sequence of $|A| + |B| - 1$ elements $a_1 + b_1 < \dots < a_{|A|} + b_1 < \dots < a_{|A|} + b_{|B|}$, all of which are contained in $A + B$, proves the claim.
\end{proof}

\begin{lemma} \label{lem:iterated-sumset-lower-bound}
Let $A_1, \dots, A_k \subseteq \Real$. Then $|A_1 + \dots + A_k| \geq |A_1| + \dots + |A_k| - k + 1$.
\end{lemma}
\begin{proof}
Apply the previous lemma $k - 1$ times.
\end{proof}

\begin{proof}[Proof of \cref{thm:subset-sum}]
The idea is to use a simple divide-and-conquer algorithm: If $A$ contains a single element $a$, then we return $\SSS(A) = \set{0, a}$ in constant time. Otherwise, we arbitrarily partition $A = A_1 \cup A_2$ into multisets $A_1, A_2$ of equal size (up to $\pm 1$). We compute $\SSS(A_1)$ and $\SSS(A_2)$ recursively, and then compute $\SSS(A) = \SSS(A_1) + \SSS(A_2)$ by one sparse sumset computation (\cref{thm:sumset}).

The correctness is immediate, but it remains to analyze the running time. Note that the recursion tree has depth $\log n$. Focus on any level $\ell \in [\log n]$. This level induces a partition into at most $2^\ell$ submultisets $A_1, \dots, A_{2^\ell} \subseteq A$. The running time of level $\ell$ is dominated by computing the sets $\SSS(A_1), \dots, \SSS(A_{2^\ell})$ near-linear output-sensitive time. By the previous lemma, this time can be bounded as follows:
\begin{equation*}
    \widetilde\Order\parens*{\sum_{i=1}^{2^\ell} |\SSS(A_i)|} \leq \widetilde\Order\parens{|\SSS(A_1) + \dots + \SSS(A_{2^\ell})| + 2^\ell} \leq \widetilde\Order(|\SSS(A)| + n) = \widetilde\Order(|\SSS(A)|).
\end{equation*}
Summing over all $\log n \leq \log |\SSS(A)|$ levels yields the claimed time bound.
\end{proof}

\paragraph{Subset Sum with a Target}
The more interesting problem is to compute all subsets sums below a prespecified target $t \in \Realnneg$, $\SSS(A, t) = \SSS(A) \cap [0, t]$. We replicate Bringmann and Nakos' result~\cite{BringmannN20} that $\SSS(A, t)$ can be computed in roughly time $|\SSS(A, t)|^{4/3}$ for the reals:

\thmsubsetsumcapped*

Having already established that (prefix-restricted) sumsets can be computed efficiently (\cref{thm:sumset,thm:prefix-sumset}), the remaining proof of \cref{thm:subsetsum-capped} does not differ much from Bringmann and Nakos' original version. We will therefore keep the presentation concise. The only issue that arises is that, in contrast to the integer setting, for the reals we cannot assume to (approximately) know~$|\SSS(A, t)|$ in advance. We deal with this issue by adding another layer of recursion (similar to the scaling trick, but somewhat more involved, see the proof of \cref{thm:subsetsum-capped}).

The general idea is to use Bringmann's two-level color-coding technique~\cite{Bringmann17}, and to use prefix-restricted sumset computations to avoid computing subset sums above the threshold $t$. Throughout, we assume that $A$ is preprocessed in such a way that~\makebox{$n \leq |\SSS(A, t)|$} (i.e., we remove elements bigger than $t$ or reduce their multiplicity appropriately).

\begin{lemma}[Second Level] \label{lem:subsetsum-second-level}
There is a Monte Carlo algorithm that, given a multiset $A \subseteq [u, 2u]$ and $t \in \Realnneg$, computes a set $S \subseteq \SSS(A, t)$ in time $\widetilde\Order(\ceil{\frac{t}{u}}^2 \cdot |\SSS(A, t)|^{4/3})$. Each subset sum $x \in \SSS(A, t)$ is contained in $S$ with probability at least $\frac12$. 
\end{lemma}
\begin{proof}
Let $k = \floor{\frac{t}{u}}$ and randomly partition $A$ into $2k^2$ subsets $A = A_1 \sqcup \dots \sqcup A_{2k^2}$. We compute the set $S := ((A_1 \cup \set{0}) + \dots + (A_{2k^2} \cup \set{0})) \cap [0, t]$ by $2k^2$ prefix-restricted sumset computations (\cref{thm:prefix-sumset}). The running time is indeed bounded by $\widetilde\Order(k^2 \cdot |\SSS(A, t)|^{4/3})$. Concerning the correctness, first note that~\makebox{$S \subseteq \mathcal S(A, t)$} as claimed. Next, observe that any fixed subset sum $x \in \SSS(A, t)$ can be expressed as the sum of at most $k$ distinct elements from $A$, say $x = a_1 + \dots + a_k$. It is easy to check that with probability at least $\frac12$ no two of these $k$ elements have been inserted into the same part $A_i$. Thus, with probability at least $\frac12$, we have $x \in S$.
\end{proof}

\begin{lemma}[First Level] \label{lem:subsetsum-first-level}
There is a Monte Carlo algorithm that, given a multiset $A \subseteq [u, 2u]$ and $t \in \Realnneg$, computes a set $S \subseteq \SSS(A, t)$ in time $\widetilde\Order(|\SSS(A, t)|^{4/3})$. Each subset sum $x \in \SSS(A, t)$ is contained in $S$ with probability at least $\frac12$.
\end{lemma}
\begin{proof}
First, if $u \leq \frac{t}{2n}$ then we can compute $\SSS(A, t) = \SSS(A)$ in near-linear time by \cref{thm:subset-sum}. Otherwise, let $k = \floor{\frac{t}{u}}$ and randomly partition $A$ into $k$ subsets $A = A_1 \sqcup \dots \sqcup A_k$. Let $t' = 12 \log k \cdot u$. For each $i \gets 1, \dots, k$, we compute a set $S_i \subseteq \SSS(A_i, t')$ by \cref{lem:subsetsum-second-level}; by repeating this algorithm for $\Order(\log k)$ times, we achieve that each subset sum $x \in \SSS(A_i, t')$ is contained in $S_i$ with probability at least $1-\frac{1}{4k}$. We finally compute $S := ((S_1 \cup \set{0}) + \dots + (S_k \cup \set{0})) \cap [0, t]$ by $k-1$ prefix-restricted sumset computations (\cref{thm:prefix-sumset}) in a binary-tree fashion. Specifically, we compute in levels~\makebox{$\ell \gets \log k, \dots, 0$} the sets $S_{\ell, 0}, \dots, S_{\ell, 2^\ell-1}$ defined by $S_{\log k, i} := S_i$ at the leaf level, and otherwise $S_{\ell, i} := (S_{\ell+1, 2i} + S_{\ell+1, 2i+1}) \cap [0, t]$, and ultimately return $S = S_{0, 0}$.

For the correctness, note that clearly $S \subseteq \SSS(A, t)$. Now fix any subset sum~\makebox{$x \in \SSS(A, t)$}. Again, $x$~can be expressed as the sum of at most $k$ distinct elements from $A$, say $x = a_1 + \dots + a_k$. By a standard argument the probability that any part $A_i$ contains more than $s := 6 \log k$ elements from~\makebox{$a_1, \dots, a_k$} is at most $k \cdot k^{-s} \cdot \binom{k}{s} \leq k \cdot (e / s)^s \leq k \cdot 2^{-s} \leq k^{-5}$. Therefore, with probability at least $1 - k^{-5} \geq \frac12$, the contribution of $a_1 + \dots + a_k$ that is mapped to the part $A_i$ is indeed contained in $S_i$, and thus $x \in S$.

We finally consider the running time. In the calls to \cref{lem:subsetsum-second-level}, note that~\smash{$\frac{t'}{u} \leq \Order(\log k)$} and therefore the overhead of $\ceil{\frac{t'}{u}}^2$ becomes negligible in the running time. For this reason, and by \cref{thm:prefix-sumset}, we can bound the total running time by~\smash{$\widetilde\Order(\sum_{\ell, i} |S_{\ell, i}|^{4/3})$}. To further bound this by~\makebox{$|\SSS(A, t)|^{4/3}$}, we distinguish two cases: On the one hand, note that for all levels $\ell \geq \ell^* := \log(24 \log k)$ we have that $S_{\ell, i} \subseteq [0, t/2]$ and therefore the capping with $[0, t]$ does not play any role. By \cref{lem:iterated-sumset-lower-bound} we can thus bound
\begin{equation*}
    \sum_{i=0}^{2^\ell-1} |S_{\ell, i}| \leq \sum_{i=0}^{2^{\ell^*}-1} |S_{\ell^*, i}| + 2^\ell \leq 2^{\ell^*} \cdot |\SSS(A, t)| + k \leq \widetilde\Order(|\SSS(A, t)|),
\end{equation*}
using that $k \leq \Order(n) \leq \Order(|\mathcal S(A, t)|)$. On the other hand, all levels $\ell < \ell^*$ consist of less than $2^{\ell^*} \leq \Order(\log k)$ computations, so we can similarly bound~\smash{$\sum_i |S_{\ell, i}| \leq \widetilde\Order(|\SSS(A, t)|)$}. All in all, we can therefore bound the total running time by $\widetilde\Order(\sum_{\ell, i} |S_{\ell, i}|^{4/3}) \leq \widetilde\Order(|\SSS(A, t)|^{4/3})$.
\end{proof}

\begin{proof}[Proof of \cref{thm:subsetsum-capped}]
We face an additional difficulty: We do not know the size $s = |\SSS(A, t)|$ in advance (which is necessary as we have to boost the calls to \cref{lem:subsetsum-first-level}). We therefore add one more recursive layer to the algorithm to approximate $s$. Specifically, we arbitrarily partition $A = A_1 \sqcup A_2$ into two halves, and recursively compute $\SSS(A_1, \frac{t}{2})$ and $\SSS(A_2, \frac{t}{2})$. We thereby obtain a good approximation of $s$ as, writing $s_1 = |\SSS(A_1, \frac{t}{2})|$ and $s_2 = |\SSS(A_2, \frac{t}{2})|$,
\begin{equation*}
    s_1 + s_2 - 1 \leq s \leq n^4 \cdot s_1^4 \cdot s_2^4.
\end{equation*}
The first inequality is due to \cref{lem:sumset-lower-bound}. For the second inequality we show that any~\makebox{$x \in \SSS(A, t)$} can be expressed as the sum of (i) at most four elements from $A \cap (\frac{t}{4}, t]$, plus (ii) at most four subset sums from $\SSS(A_1, \frac{t}{2})$, plus (iii) at most four subset sums from $\SSS(A_2, \frac{t}{2})$. To see this, note that we can express $x = x_0 + x_1 + x_2$, where (i) $x_0$ is the sum of at most four elements from $A \cap (\frac{t}{4}, t]$, and (ii) $x_1 \in \SSS(A_1 \cap [0, \frac{t}{4}], t)$ and (iii) $x_2 \in \SSS(A_2 \cap [0, \frac{t}{4}])$. Next, express $x_1 = a_1 + \dots + a_\ell$ for some items~\makebox{$a_1, \dots, a_\ell \in A_1 \cap [0, \frac{t}{4}]$}. We peel of the largest prefix sum~\makebox{$x_{1,1} = a_1 + \dots + a_j \leq \frac{t}{2}$}, and remark that $x_{1, 1} \geq \frac{t}{4}$. We similarly peel of at most three other prefix sums $x_{1, 2}, x_{1, 3}, x_{1, 4} \in \SSS(A_1, \frac{t}{2})$ to express $x_1 = x_{1, 1} + x_{1, 2} + x_{1, 3} + x_{1, 4}$. A similar argument shows that $x_2$ can be expressed as the sum of at most four elements from $\SSS(A_2, \frac{t}{2})$.

We proceed as follows. Let $A^{(\ell)} = A \cap (\frac{t}{2^{\ell+1}}, \frac{t}{2^\ell}]$ for $0 \leq \ell \leq L := \ceil{\log n}$ and $A^{(> L)} = A \cap [0, \frac{t}{2^L}]$, and note that~\smash{$A^{(0)} \sqcup \dots \sqcup A^{(L)} \sqcup A^{(> L)}$} forms a partition of $A$. We compute~\smash{$\SSS(A^{(> L)}, t) = \SSS(A^{(> L)})$} in near-linear time by \cref{thm:subset-sum}. We compute all other sets $\SSS(A^{(\ell)}, t)$ by calls to \cref{lem:subsetsum-first-level}. By repeating each call $\log(n^{102} \cdot s_1^2 \cdot s_2^2)$ times, say, we can compute the correct set $\SSS(A^{(\ell)}, t)$ with high probability $1 - \frac{1}{n^{101}}$. We combine all sets using $L$ prefix-restricted sumset computations (\cref{thm:prefix-sumset}) to~\smash{$\SSS(A, t) = \SSS(A^{(0)}, t) + \dots + \SSS(A^{(L)}, t) + \SSS(A^{(> L)}, t)$}. Taking a union bound over the $n$ recursive calls, the entire algorithm succeeds with high probability $1 - \frac{1}{n^{100}}$. Moreover, similar to before the running time can be bounded by~\smash{$\widetilde\Order(|\SSS(A, t)|^{4/3})$}, using the fact that in each recursive level the total size of all sets of subset sums is upper bounded by $|\SSS(A, t)|$ (by \cref{lem:iterated-sumset-lower-bound}).
\end{proof}

\bibliographystyle{plainurl}
\bibliography{paper}

\end{document}